\documentclass[10pt]{amsart}

\usepackage{enumerate}

\newtheorem{theorem}{Theorem}[section]

\newtheorem{lemma}[theorem]{Lemma}

\newtheorem{remark}[theorem]{Remark}
\newtheorem{example}[theorem]{Example}

\begin{document}

\title[Exponential utility maximization under model uncertainty]{
Exponential utility maximization under model uncertainty for unbounded endowments}
\author[Daniel Bartl]{Daniel Bartl$^*$}	
\thanks{$^*$Department of Mathematics, University of Konstanz, daniel.bartl@uni-konstanz.de\\
The author would like to thank Michael Kupper and two anonymous referees for valuable comments.	
The author is supported by the Austrian Science Fund (FWF) under grant Y00782.}

\date{\today}

\begin{abstract} 
	We consider the robust exponential utility maximization problem in discrete time:
	An investor maximizes the worst case expected exponential utility with respect to
	a family of nondominated probabilistic models of her endowment by dynamically
	investing in a financial market, and statically in available options.
	
	We show that, for any measurable random endowment
	(regardless of whether the problem is finite or not)
	an optimal strategy exists, a dual representation in terms of (calibrated) martingale
	measures holds true, and that the problem satisfies the dynamic programming principle
	(in case of no options).
	Further it is shown that the value of the utility maximization problem converges
	to the robust superhedging price as the risk aversion parameter gets large, 
	and examples of nondominated probabilistic models are discussed.
\end{abstract}

\maketitle
\setcounter{equation}{0}

\section{Introduction}

In this article we study the problem of robust exponential utility maximization
in discrete time. 
Here the term robust reflects uncertainty about the true probabilistic model 
and the consideration of a whole family of models as a consequence.
This is not a new concept and since the seminal papers 
\cite{gilboa1989maxmin} and \cite{maccheroni2006ambiguity}
it has gained a lot of attention,
see  e.g.~\cite{acciaio2013model,beiglbock2013model,beiglbock2015complete,bouchard2015arbitrage,burzoni2016pointwise,cheridito2015representation,denis2013optimal,hobson1998robust,neufeld2015robust,nutz2014utility,peng2007g}
and \cite{cheridito2016duality} for an overview.
To state our problem more precisely, given the exponential utility function
\[ U(x):=-\exp(-\gamma x) \]
with risk-aversion parameter $\gamma>0$, a possibly nondominated set of probabilistic models
$\mathcal{P}$ and the agent's random endowment $X$,
we are interested in the optimization problem 
\begin{align}
\label{eq:intro.primal}
\sup_{(\vartheta,\alpha)\in\Theta\times\mathbb{R}^e}\inf_{P\in\mathcal{P}} E_P[ U(X+(\vartheta\cdot S)_T+\alpha(g-g_0))].
\end{align}
Here $g^1,\dots,g^e$ are traded options available for buying and selling at time 0 for the prices $g_0^1,\dots,g_0^e$, the set $\Theta$ consists of all predictable dynamic trading strategies for the (discounted) stock $S$,
and $(\vartheta\cdot S)_T+\alpha(g-g_0)$ is the outcome of a semistatic trading strategy 
$(\vartheta,\alpha)\in\Theta\times\mathbb{R}^e$.

The first immediate question when investigating the optimization problem \eqref{eq:intro.primal}
is whether an optimal strategy $(\vartheta,\alpha)$ 
(which should be defined simultaneously under all models $P\in\mathcal{P}$) exits.
Due to the absence of a measure capturing all zero sets and the failure of classic
arguments such as Komlos' theorem as a consequence, this is nontrivial.
Our second interest lies in the validity of a dual representation with respect to linear pricing measures,
namely if \eqref{eq:intro.primal} is equal to
\[ -\exp\big(-\inf_{Q\in\mathcal{M}} \big( \gamma E_Q[X] + H(Q,\mathcal{P}) \big) \big),\]
where $\mathcal{M}$ denotes the set of all martingale measures $Q$ for the stock $S$ calibrated to the options
(i.e.~$E_Q[g^i]=g^i_0$ for $1\leq i\leq e$)  under which the robust entropy $H(Q,\mathcal{P})$ is finite.
Finally we study if \eqref{eq:intro.primal} satisfies the dynamic programming
principle (in case without options), 
meaning that it is possible to analyze the problem locally
and later ``glue" everything together. In particular this implies that a strategy
which is optimal at time 0, will be optimal again, if one starts to solve the optimization
problem at some positive time $t$.

\vspace{0.5em}

The main contribution of this paper is to show that positive answers to all three questions,
namely the existence of an optimal strategy, duality, and dynamic programming 
can be given under weak assumptions, see Theorem \ref{thm:main} and Theorem \ref{thm:main.options}.
Further, it is shown that a scaled version of \eqref{eq:intro.primal} converges to the 
minimal superhedging price of $X$ if the risk-aversion parameter $\gamma$ tends to infinity,
see Theorem \ref{thm:limit.superhedg}.
In fact, we adopt the setting suggested by Bouchard and Nutz in 
the milestone paper \cite{bouchard2015arbitrage} 
and show by means of optimal control, that for any unbounded measurable (lower semianalytic) 
random endowment $X$ (regardless of whether the optimization problem \eqref{eq:intro.primal} 
is finite or not), existence, duality, and the dynamic programming principle hold true.

\vspace{0.5em}

Needless to say, utility maximization is an important topic in mathematical finance
starting with \cite{kramkov1999asymptotic,merton1971optimum}.
In case of exponential utility function 
(though in a continuous-time and non-robust setting),
\cite{frittelli2000minimal} and \cite{delbaen2002exponential} were the first to prove
duality and existence, which lead to further analysis, for example 
a BSDE characterization of the optimal value and solution in an incomplete 
market and under trading constraints is given in \cite{hu2005utility},
and the dynamics and asymptotics in the risk-aversion parameter $\gamma$
are studied in \cite{mania2005dynamic}.
In the presence of uncertainty, starting with \cite{quenez2004optimal} and \cite{schied2004risk}, 
most results are obtained under the assumption that
$\mathcal{P}$ is dominated, see e.g.~\cite{backhoff2015robust,gundel2005utility,owari2011robust}.
The literature focusing on a nondominated set $\mathcal{P}$ is still comprehensible and in
continuous-time results are given in \cite{denis2013optimal,matoussi2015robust,neufeld2015robust}.

In the present setting (that is discrete-time and a nondominated set $\mathcal{P}$), 
the dynamic programming principle and the existence of an optimal strategy are first shown in \cite{nutz2014utility},
where the author considers a random utility function $U$ defined on $\Omega\times\mathbb{R}_+$ 
satisfying a certain boundedness 
(which would correspond to a random endowment that is bounded from below in our setting).
More recently, there are three papers generalizing the result of \cite{nutz2014utility}.
In \cite{blanchard2016robust}, the boundedness of the random utility 
(still defined on the positive real line) 
is replaced by a certain integrability condition and dynamic programming as well as the
existence of an optimal strategy is shown. 
In \cite{neufeld2016robust,neufeld2017nonconcave}, the random utility function
(which may be non-concave in the second work) is no longer defined on the positive real line,
but satisfies certain boundedness similar to \cite{nutz2014utility}.
Moreover, the market is more general and includes e.g.~trading constraints or proportional transaction cost.
Convergence of the utility indifference prices (to the superheding price) is shown in 
\cite{blanchard2017convergence}.
Duality on the other hand is shown in Section 4.2 of \cite{cheridito2016duality}
under a compactness condition on the set $\mathcal{P}$ and (semi-)continuity of the
random endowment $X$.

\vspace{0.5em}

In order to lighten notation, we will assume without loss of generality that the prices
of the traded options are 0 and, instead of \eqref{eq:intro.primal}, consider the equivalent problem
\[ \inf_{(\vartheta,\alpha)\in\Theta\times\mathbb{R}^e}
\sup_{P\in\mathcal{P}} \log E_P[\exp(X+(\vartheta\cdot S)_T + \alpha g)].\]
It is clear that both problems are one to one, except all results for an 
endowment $X$ in the original problem hold for $-X$ in the transformed one, and vice versa.

The remainder of this paper is organized as follows: 
Section \ref{sec:main} contains the setting, all main results, 
a discussion of the assumptions, and some examples. 
Sections \ref{sec:one.period} and Section \ref{sec:multi.period} are devoted to the proofs for the one period and general case, respectively. 
Finally, technical proofs are given in Appendix \ref{sec:app.proofs} and a brief introduction
to the theory of analytic sets is given in Appendix \ref{sec:app.analytic}.

\section{Main results}
\label{sec:main}

\subsection{Setting}

Up to a change regarding the no-arbitrage condition 
(discussed in Remark \ref{rem:main.discussion}), we work in the setting 
proposed by Bouchard and Nutz \cite{bouchard2015arbitrage}, which is briefly summarized below.
Analytic sets and the general terminology are shortly discussed in Appendix \ref{sec:app.analytic}.
Let $\Omega_0$ be a singleton and $\Omega_1$ be a Polish space. 
Fix $d,T\in\mathbb{N}$, let $\Omega_t:=\Omega_1^t$, and define
$\mathcal{F}_t$ to be the universal completion of the Borel $\sigma$-field on $\Omega_t$
for each $0\leq t\leq T$.
To simplify notation, we denote $(\Omega,\mathcal{F})=(\Omega_T,\mathcal{F}_T)$
and often consider $\Omega_t$ as a subset of $\Omega$.
For $s<t$, some fixed $\omega\in\Omega_s$, and a function $X$ with domain $\Omega_t$, we consider
$X(\omega,\cdot)$ as a function with domain $\Omega_{t-s}$, i.e.~$\Omega_{t-s}\ni\omega'\mapsto X(\omega,\omega')$.
For each $0\leq t\leq T-1$ and $\omega\in\Omega_t$, there is a given convex and nonempty
set of probabilities $\mathcal{P}_t(\omega)\subset\mathfrak{P}(\Omega_1)$, 
which can be seen as all possible probability scenarios 
for the price of the stock at time $t+1$, given the history $\omega$.
The assumption throughout is that the stock $S_t\colon\Omega_t\to\mathbb{R}^d$ is Borel
and that the set-valued mapping $\mathcal{P}_t$ has analytic graph.
The latter in particular ensures that 
\begin{align}
\label{eq:P.time.consistent}
 \mathcal{P}:=\{ P=P_0\otimes\cdots\otimes P_{T-1} : P_t(\cdot)\in\mathcal{P}_t(\cdot) \}
\end{align}
is not empty.
Here, each $P_t$ is a selector of $\mathcal{P}_t$ , i.e.~a universally measurable function
$P_t\colon\Omega_t\to\mathfrak{P}(\Omega_1)$ satisfying $P_t(\omega)\in\mathcal{P}_t(\omega)$
for each $\omega$, and the probability $P$ on $\Omega$ is defined by
$P(A):=\int_{\Omega_1}\cdots\int_{\Omega_1} 1_A(\omega_1,\dots,\omega_T)
P_{T-1}(\omega_1,\dots,\omega_{T-1},d\omega_T)\cdots P_0(d\omega_1)$.
The set of all dynamic trading strategies is denoted by $\Theta$ and an element $\vartheta\in\Theta$
is a vector $\vartheta=(\vartheta_1,\cdots,\vartheta_T)$ consisting of $\mathcal{F}_{t-1}$-measurable
mappings $\vartheta_t\colon\Omega_{t-1}\to\mathbb{R}^d$. 
The outcome at time $t$ of trading according to the dynamic strategy $\vartheta$ starting at time $s\leq t$ is given by
\[ (\vartheta\cdot S)_s^t:=\vartheta_{s+1}\Delta S_{s+1} +\cdots + \vartheta_t\Delta S_t,\quad
\text{where}\quad \Delta S_u:=S_u-S_{u-1} \]
and $\vartheta_u\Delta S_u:=\sum_{i=1}^d \vartheta_u^i\Delta S_u^i$ is the inner product.
As $\mathcal{P}$ has a dynamic form, one can consider both a local and a global no arbitrage condition:
The global $\mathrm{\mathop{NA}}(\mathcal{P})$ condition is satisfied if $(\vartheta\cdot S)_0^T\geq 0$ $\mathcal{P}$-q.s.~implies 
$(\vartheta\cdot S)_0^T= 0$ $\mathcal{P}$-q.s.~for every $\vartheta\in\Theta$, 
and the local NA($\mathcal{P}_t(\omega)$) condition (for fixed $0\leq t\leq T-1$ and $\omega\in\Omega$)
is satisfied if $h\Delta S_{t+1}(\omega,\cdot)\geq 0$
$\mathcal{P}_t(\omega)$-q.s.~implies $h\Delta S_{t+1}(\omega,\cdot)= 0$ 
$\mathcal{P}_t(\omega)$-q.s.~for every $h\in\mathbb{R}^d$.	
Throughout this article, we assume that 
\begin{align}
\label{eq:no.arbitrage}
\mathrm{\mathop{NA}}(\mathcal{P}_t(\omega)) \text{ holds for every } 0\leq t \leq T-1
\text{ and }\omega\in\Omega_t.
\end{align}
Note that this assumption is purely technical, 
as $\mathrm{\mathop{NA}}(\mathcal{P})$ holds true if and only if the set of all 
$\omega$ such that $\mathrm{\mathop{NA}}(\mathcal{P}_t(\omega))$ fails for some $t$ is 
a zero set under all $P\in\mathcal{P}$, see \cite[Theorem 4.5]{bouchard2015arbitrage}
and Remark \ref{rem:main.discussion} for a discussion.
Finally, define
\[ \mathcal{M}
=\big\{ Q\in\mathfrak{P}(\Omega): 
S \text{ is a martingale under $Q$ and } H(Q,\mathcal{P})<+\infty\big\}, \]
to be the set of martingale measures with finite robust entropy 
\[ H(Q,\mathcal{P}):=\inf_{P\in\mathcal{P}} H(Q,P)
\quad\text{where}\quad
H(Q,P):=\begin{cases} 
E_P\big[\frac{dQ}{dP}\log\frac{dQ}{dP}\big]&\text{if }Q\ll P,\\
+\infty &\text{else.}
\end{cases} \]
Throughout the convention $E_P[X]:=E_P[X^+]-E_P[X^-]$ with $E_P[X]:=-\infty$ if $E_P[X^-]=+\infty$ is in force; 
in particular $X$ is integrable with respect to $P$ if and only if $E_P[X]\in\mathbb{R}$.

\subsection{Main results}

\begin{theorem}[Without options]
\label{thm:main}
	Let $X\colon\Omega\to(-\infty,+\infty]$ be upper semianalytic.
	Then 
	\begin{align}
	\label{eq:optim.problem}
	\inf_{\vartheta\in\Theta}\sup_{P\in\mathcal{P}}
	\log E_P\big[\exp\big(X + (\vartheta\cdot S)_0^T\big)\big]
	=\sup_{Q\in\mathcal{M}}	\big(E_Q[X] -H(Q,\mathcal{P}) \big)
	\end{align}
	and both terms are not equal to $-\infty$.
	Moreover, the infimum over 	$\vartheta\in\Theta$ is attained and the optimization
	problem satisfies the dynamic programming principle;
	see Theorem \ref{thm:multiperiod} for the precise formulation of the last statement.
\end{theorem}

In addition to the previous setting, assume that there are $e\in\mathbb{N}\cup\{0\}$ options
($e=0$ corresponding to the case without options)
i.e.~Borel functions $g^1,\dots,g^e\colon\Omega\to\mathbb{R}$, 
available at time $t=0$ for price zero.
The outcome of a semistatic trading strategy $(\vartheta,\alpha)\in\Theta\times\mathbb{R}^e$ 
equals $(\vartheta\cdot S)_0^T+ \alpha g$, where $\alpha g:=\sum_{i=1}^e \alpha_ig^i$ again denotes the inner product.
In addition to the already imposed no arbitrage condition, assume that 
$(\vartheta\cdot S)_0^T+ \alpha g\geq 0$ $\mathcal{P}$-q.s.~implies
$(\vartheta\cdot S)_0^T+ \alpha g= 0$ $\mathcal{P}$-q.s.~for every strategy $(\vartheta,\alpha)\in\Theta\times\mathbb{R}^e$.

\begin{theorem}[With options]
\label{thm:main.options}
	Fix a Borel function $Z\colon\Omega\to[0,+\infty)$ such that $|g^i|\leq Z$ for $1\leq i\leq e$
	and let $X\colon\Omega\to\mathbb{R}$ be an upper semianalytic function satisfying $|X|\leq Z$.
	Then it holds
	\begin{align*}
	\inf_{(\vartheta,\alpha)\in\Theta\times\mathbb{R}^e}\sup_{P\in\mathcal{P}}
	\log E_P\big[\exp\big(X + (\vartheta\cdot S)_0^T + \alpha 	g\big)\big]
	=\sup_{Q\in\mathcal{M}_g}	\big(E_Q[X] -H(Q,\mathcal{P}) \big),
	\end{align*}
	where $\mathcal{M}_g$ denotes the set of all $Q\in\mathcal{M}$ 
	with $E_Q[Z]<+\infty$ and $E_Q[g^i]=0$ for $1\leq i\leq e$.
	Moreover, the infimum over $(\vartheta,\alpha)\in\Theta\times\mathbb{R}^e$ is attained.
\end{theorem}

\begin{remark}
\label{rem:main.assumptions}
{\rule{0mm}{1mm}\\[-3.25ex]\rule{0mm}{1mm}}
\begin{enumerate}[1)]
	\item 
	The no-arbitrage condition {\rm NA}$(\mathcal{P})$ is essential.
	Indeed, even if both sides in \eqref{eq:optim.problem} do not take the value 
	$-\infty$, the condition {\rm NA}$(\mathcal{P})$ does not need to hold -- 
	nor does equation \eqref{eq:optim.problem}; see Appendix \ref{sec:app.proofs}.
	\item
	If $X$ is allowed to take the value $-\infty$ in Theorem \ref{thm:main}, 
	then neither duality nor the existence of an optimal strategy hold true;
	see Appendix \ref{sec:app.proofs}.
	\item 
	In general, due the supremum over $P\in\mathcal{P}$,
	the minimizer $\vartheta$ in \eqref{eq:optim.problem} is not unique
	and the supremum over $Q$ is not attained.
	\item In Theorem, \ref{thm:main} the set $\mathcal{M}$ can be replaced by
	$\mathcal{M}(Y):=\{Q\in\mathcal{M} : E_Q[Y]<+\infty\}$, where $Y\colon\Omega\to[0,+\infty)$
	is an arbitrary function such that $-Y$ is upper semianalytic.
	The same holds true for Theorem \ref{thm:main.options}, i.e.~one can replace
	$\mathcal{M}_g$ by $\mathcal{M}_g(Y):=\{Q\in\mathcal{M}_g : E_Q[Y]<+\infty\}$.
\end{enumerate}
\end{remark}

Another interesting problem is the study of asymptotic behavior of the optimization problem
in the risk-aversion parameter $\gamma$, see e.g.~\cite{mania2005dynamic}.
Let us give some motivation: Typically the superhedging price 
\[\pi(X):=\inf\{ m\in\mathbb{R} : m+(\vartheta\cdot S)_0^T +\alpha g \geq X\,
\mathcal{P}\text{-q.s.~for some } (\vartheta,\alpha)\in\Theta\times\mathbb{R}^e \}\]
is extremely high.
A natural way of shrinking $\pi$ is to allow $m+(\vartheta\cdot S)_0^T +ug<X$ with positive
probability in a ``controlled'' way, see e.g.~\cite{cheridito2016duality,follmer1999quantile}. 
More precisely, define
\[ \pi_\gamma(X):=\inf\Big\{ m \in\mathbb{R}:
\begin{array}{l}
\sup_{P\in\mathcal{P}}\frac{1}{\gamma}\log 
E_P[\exp(\gamma(X-m-(\vartheta\cdot S)_0^T-\alpha g))]\leq0\\
\text{for some } (\vartheta,\alpha)\in\Theta\times\mathbb{R}^e
\end{array}\Big\} \]
for each risk-aversion parameter $\gamma>0$.
Then $\pi_\gamma(X)\leq\pi(X)$ by definition and since $\exp(\gamma x)/\gamma\to+\infty1_{(0,+\infty]}(x)$
as $\gamma\to+\infty$, an evident question is whether the same holds true for the superhedging prices.

\begin{theorem}[Entropic hedging]
\label{thm:limit.superhedg}
	In the setting of Theorem \ref{thm:main.options}, it holds
	\[ \pi(X)=\lim_{\gamma\to+\infty}\pi_\gamma(X) \]
	and the limit in $\gamma$ is a supremum over $\gamma>0$.
\end{theorem}

\begin{remark}
\label{rem:main.discussion}
	The reason why we require $\mathrm{NA}(\mathcal{P}_t(\omega))$ to hold for every $\omega$
	and not only for $\mathcal{P}$-quasi every $\omega$ as in Bouchard and Nutz
	\cite{bouchard2015arbitrage}, is the following.
	In order to apply the one period results (i.e.~duality and existence of an optimal strategy)
	to the local problem $\mathcal{E}_t(\omega,x)$ 
	(see \eqref{eq:local.optimization} for the precise definition), one needs that 
	$\mathcal{E}_{t+1}(\omega\otimes_t\omega',x)>-\infty$ for $\mathcal{P}_t(\omega)$-quasi all $\omega'\in\Omega_1$, see point 2) in Remark \ref{rem:main.assumptions}.
	However, to ensure the latter, $\mathrm{NA}(\mathcal{P}_{t+1}(\omega\otimes_t\omega'))$ 
	needs to holds for $\mathcal{P}_t(\omega)$-quasi all $\omega'\in\Omega_1$.
	Due to the fact that the set
	$N_{t+1}:=\{ \tilde{\omega}\in\Omega_{t+1} :  \mathrm{NA}(\mathcal{P}_t(\tilde{\omega}))\text{ fails}\}$ 
	is merely universally measurable, 
	it is not clear that this condition holds true for ``sufficiently many'' $\omega\in\Omega_t$.

	In case of only one measure 
	(i.e.~$\mathcal{P}=\{P=P_0\otimes\cdots\otimes P_{T-1}\}$ is a singleton), 
	this problem has an easy solution: Since 
	$P(N_t)=0$ for every $t$, one can redefine $P_t$ by 
	$\tilde{P}_t:=P_t1_{N_t^c}+\delta_{S_t} 1_{N_t}$. Then
	$P=\tilde{P}_0\otimes\cdots\otimes \tilde{P}_{T-1}$
	and $\tilde{\mathcal{P}}_t(\omega):=\{\tilde{P}_t(\omega)\}$ has analytic graph
	(a proof is given in Appendix \ref{sec:app.proofs}).
	The same can be done in the setting of general $\mathcal{P}$,
	as long as one requires the sets $N_t$ to be Borel; otherwise the graph of 
	$\tilde{\mathcal{P}}_t:=\mathcal{P}_t1_{N_t^c}+\{\delta_{S_t}\}1_{N_t}$
	needs not to be analytic.

	Finally note that $\mathcal{P}$ is defined though the sets $\mathcal{P}_t$ 
	(and not the other way around), which means that the assumption 
	for $\mathrm{NA}(\mathcal{P}_t(\omega))$ to hold for every $\omega$ 
	does not seem restrictive regarding applications; see also Section \ref{sec:examples}.

	Recently, using a different approach, \cite{deng2018utility} were able to get rid of this assumption
\end{remark}		

\begin{remark}
	The (technical but crucial) assumption that the graph of $\mathcal{P}_t$ is analytic has two consequences: 
	It allows for measurable selection arguments and enables to define pointwise
	conditional sublinear expectations, i.e.~ensure that 
	\begin{align}
	\label{eq:sublin.expectation}
	\mathcal{E}(X|\Omega_t)(\omega):=\sup_{P\in\mathcal{P}_t(\omega)} E_P[X(\omega,\cdot)]
	\quad\text{for } \omega\in\Omega_t\text{ and } X\colon\Omega_{t+1}\to\mathbb{R}
	\end{align}
	is upper semianalytic as a mapping of $\omega$ whenever $X$ is
	\cite{bouchard2015arbitrage,nutz2013constructing}.
	The converse holds true as well: Given an arbitrary sublinear conditional expectation
	$\mathcal{E}(\cdot|\Omega_t)$ (satisfying some continuity), there always exists 
	a set-valued mapping $\mathcal{P}_t$
	with analytic graph	such that \eqref{eq:sublin.expectation}	holds true
	\cite[Theorem 1.1]{bartl2016conditional}.
	Similarly, the ``time-consistency'' \eqref{eq:P.time.consistent} of $\mathcal{P}$ is 
	equivalent to the tower-property \cite[Theorem 1.2]{bartl2016conditional}. 
\end{remark}

\subsection{Examples}
\label{sec:examples}
In this section, we discuss a general method of robustifying a given probability 
and also give applications to financial models.
All nontrivial claims are proven at the beginning of Appendix \ref{sec:app.proofs}.

In many cases, the physical measure is not known a priori, but rather a result of collecting data
and estimation. In particular, the estimator is not equal to, but only ``converges'' (as the data grow richer)
to the actual unknown physical measure. A canonical way of taking this into account therefore consists of adding 
some sort of ``neighborhood'' to the estimator $P^\ast=P^\ast_0\otimes\cdots\otimes P_{T-1}^\ast$,
i.e.~to define 
\begin{align}
\label{eq:P.robust.general}
\mathcal{P}_t(\omega):=\{ P\in\mathfrak{P}(\Omega_1) : 
\mathop{\mathrm{dist}}(P,P_t^\ast(\omega))\leq\varepsilon_t(\omega) \}.
\end{align}
Here, as the name suggests,
\[\mathop{\mathrm{dist}}\colon\Omega_1\times\Omega_1\to[0,+\infty] \]
can be thought of a distance and $\varepsilon_t\colon\Omega_t\to[0,+\infty]$
as the size of the neighborhood.
If $\mathop{\mathrm{dist}}$, $\varepsilon_t$ (and $P^\ast_t$) are Borel
-- from now on a standing assumption -- then $\mathcal{P}_t$ 
has analytic graph.
If $\mathop{\mathrm{dist}}$ is in fact a metric
or at least fulfills $\mathop{\mathrm{dist}}(P,P)=0$, the values of $\mathcal{P}_t$ are also nonempty.
Since the distance should be compatible with estimation, 
natural choices include the Wasserstein distances of order $p$ or, more generally, 
the cost of transportation i.e.~
\[ \mathop{\mathrm{dist}}(Q,P):=\inf\Big\{ \int_{\Omega_1\times\Omega_1} c(x,y)\Pi(dx,dy) : \Pi \Big\} \]
where the infimum is taken over all measures on the product $\Pi\in\mathfrak{P}(\Omega_1\times\Omega_1)$
with $\Pi(\cdot\times\Omega_1)=Q$  and $\Pi(\Omega_1\times\cdot)=P$, and
$c\colon\Omega_1\times\Omega_1\to[0,+\infty]$ is a given lower semicontinuous function (the ``cost'').
This includes the Wasserstein distance of order $p$; then the cost $c$ equals a metric
on $\Omega_1$ to the power $p$; see e.g.~Chapter 5 and 6 in \cite{villani2008optimal}.
This traceable distance has many advantages, e.g.~that besides metrizing weak convergence, 
it controls the integrability of the tails. 
In this case, $\mathcal{P}_t$ has convex values.
Moreover, \cite{bartl2018computational} provide a finite dimensional formula how to compute the worst case expectation over all probabilities in a Wasserstein neighborhood of a baseline distribution.

The above method can also be applied when a certain model for the dynamics of the underlying is fixed
and only the parameters are uncertain. For simplicity assume 
that $\Omega=\mathbb{R}^T$, $S_t(\omega)=\omega_t$
is the canonical space of a one-dimensional stock.
We illustrate in two concrete examples: the Binomial model, which reads as
\begin{align}
\label{eq:dyn.S.binomal}
S_{t+1}(\omega,\cdot)=S_t(\omega) +B(\cdot)
\end{align}
for every $t$ and $\omega\in\Omega_t$ where $B\colon\Omega_1\to\mathbb{R}$ is binomially distributed,
and a discrete version of the Black-Scholes model, which reads as
\begin{align}
\label{eq:dyn.S.black.scholes}
S_{t+1}(\omega,\cdot)=S_t(\omega)\big(\mu \Delta t + \sigma\Delta W(\cdot)\big)
\end{align}
where $\mu\in\mathbb{R}$, $\sigma,\Delta t>0$, and $\Delta W\colon\Omega_1\to\mathbb{R}$ is normally distributed
with mean 0 and variance $\Delta t$; we write $\Delta W\sim N(0,\Delta t)$.
Defining $f_t(\omega,x):=S_t(\omega)+x$, $X:=B$ in case of the Binomial, and 
$f_t(\omega,x):=S_t(\omega)x$, $X:=\mu\Delta t+ \sigma \Delta W$ in case of the Black-Scholes model, 
it follows that both can be written in the more general form
\begin{align}
\label{eq:dyn.S.general}
S_{t+1}(\omega,\cdot)=f_t(\omega, X(\cdot)),
\end{align}
where $f_t\colon\Omega_t\times\mathbb{R}\to\mathbb{R}$ and $X\colon\Omega_1\to\mathbb{R}$ are Borel.
In terms of distributions, \eqref{eq:dyn.S.general} means nothing but
\[ \mathop{\mathrm{law}} S_{t+1}(\omega,\cdot) = R \circ f(\omega,\cdot)^{-1},
\quad \text{where } R:=\mathop{\mathrm{law}} X. \]
Therefore, a canonical way of robustifying a given model of the form 
\eqref{eq:dyn.S.general} is to replace $R$ in the equation above by a set
$\mathcal{R}_t(\omega)\subset\mathfrak{P}(\Omega_1)$, and to define
\[\hat{\mathcal{P}}_t(\omega)
:=\{\mathop{\mathrm{law}} S_{t+1}(\omega,\cdot) = R \circ f(\omega,\cdot)^{-1}: 
 R\in\mathcal{R}_t(\omega)\}.\]
For example, in line with the first part of this section, one can take some neighborhood
\begin{align}
\label{eq:Rt.dist} 
 \mathcal{R}_t(\omega)=\{ R\in\mathfrak{P}(\Omega_1) : 
\mathop{\mathrm{dist}}(R,\mathop{\mathrm{law}} X)\leq\varepsilon_t(\omega) \},
\end{align}
or, if there are even less data, one might argue that 
\begin{align}
\label{eq:Rt.phitn} 
\mathcal{R}_t(\omega):=\{ R\in\mathfrak{P}(\Omega_1) : E_R[\phi_t^i(\omega, \cdot)]\leq 1
\text{ for } 1\leq i\leq n \}
\end{align}
for some given Borel functions $\phi_t^1,\dots,\phi_t^n\colon \Omega_t\times\mathbb{R}^d\to\mathbb{R}$
is a good choice.
Here, if $\inf_{x} \phi_t^i(\omega,x)\leq 0$ for all $i$ and 
$f_t$ in \eqref{eq:dyn.S.general} is such that 
$S_t(\omega)$ lies in the relative interior of $f_t(\omega,\mathbb{R}^d)$ 
for every $\omega\in\Omega_t$ -- an assumption which is usually fulfilled --
then the resulting model of \eqref{eq:Rt.phitn} satisfies NA$(\mathcal{P}_t(\omega))$ for every $t$ and $\omega$.
The same holds true for $\mathcal{R}_t$ defined by \eqref{eq:Rt.dist} under the mentioned
assumption on $f_t$ if e.g.~$\mathop{\mathrm{dist}}$ is the Wasserstein distance of order $p$
and $X$ has a finite $p$-th moment.

On a technical level, $\mathcal{R}_t$ defined by \eqref{eq:Rt.phitn} has analytic graph
and so do $\hat{\mathcal{P}}_t$ and $\mathcal{P}_t$, the latter begin defined as
$\mathcal{P}_t(\omega):=\mathop{\mathrm{conv}} \hat{\mathcal{P}}_t(\omega)$
the convex hull of $\hat{\mathcal{P}}_t$.
The same holds true for $\mathcal{R}_t$ defined by \eqref{eq:Rt.dist}.

\begin{example}[Binomial model]
	Besides what was mentioned above, another natural generalization of the Binomal model 
	is to allow for the jump size and probability 
	to take values in some intervals (which may depend on the time $t$ and past $\omega\in\Omega_t$).
	This corresponds to 
	\[ \mathcal{R}_t(\omega):=\big\{ p\delta_a +(1-p)\delta_b : p\in[\underline{p}_t(\omega),\overline{p}_t(\omega)],
	a\in[\underline{a}_t(\omega), \overline{a}_t(\omega)], b\in[\underline{b}_t(\omega),\overline{b}_t(\omega)] 
	\big\}\]
	where $0<\underline{p}_t\leq \overline{p}_t<1$, 
	$\underline{a}_t\leq \overline{a}_t<0<\underline{b}_t\leq \overline{b}_t$ are Borel functions.
	Here $\delta_a$ denotes the Dirac measure at point $a$.
	Note that {\rm NA}$(\mathcal{P}_t(\omega))$ is trivially satisfied
	for every $t$ and $\omega$.
\end{example}

Regarding the Black-Scholes model in continuous time, there is a popular and well-studied 
way of robustification, see e.g.~\cite{peng2007g}:
Consider all models \eqref{eq:dyn.S.black.scholes} with $\mu$ and volatility $\sigma$ in some given intervals.
This can be done as in the previous example, however, then each $\mathcal{P}_t(\omega)$ and, therefore,
also the resulting family $\mathcal{P}$ is dominated (by the Lebesgue measure).
In the present discrete-time setting, it seems more interesting to discard the assumption of 
normality of $\Delta W$ in \eqref{eq:dyn.S.black.scholes}.

\begin{example}[Black-Scholes]
	Fix two Borel functions 
	$\mu_t\colon\Omega_1\to\mathbb{R}$ and $\sigma_t\colon\Omega_1\to(0,+\infty)$,
	and let $\varepsilon_t$ and $\mathop{\mathrm{dist}}$ be as above.
	Now define  
	\[ \mathcal{R}_t(\omega):=\big\{ R\ast \delta_{\mu\Delta t} : 
	\mu\in[\underline{\mu_t}(\omega),\overline{\mu}_t(\omega)]
	\text{ and } \mathop{\mathrm{dist}}(R,N(0,\sigma_t^2(\omega)\Delta t))\leq\varepsilon_t(\omega)  \big\},\]
	where $R\ast \delta_{\mu\Delta t}$ denotes the convolution $R\ast \delta_{\mu\Delta t}(A):=R(A-\mu\Delta t)$.
	The set $\hat{\mathcal{P}}_t$ therefore corresponds to the Black-Scholes model with drift and 
	volatility uncertainty in the sense that one considers all models
	\[ S_{t+1}(\omega,\cdot)=S_t(\omega)\big(\mu \Delta t + Y\big),
	\quad 
	\begin{array}{l}
	\text{$\mu\in[\underline{\mu_t}(\omega),\overline{\mu}_t(\omega)]$ and the law of}\\
	\text{$Y$ is $\varepsilon_t(\omega)$ close to $N(0,\sigma_t^2(\omega)\Delta t)$}
	\end{array}\]
	simultaneously.
	To be more in line with the original model, one can also require that $R$ (resp.~$Y$) has mean 0
	in the definition of $\mathcal{R}_t$.
	Note that for any reasonable choice for the distance (e.g.~Wasserstein),
	the set $\mathcal{P}_t(\omega)$ satisfies all of our assumptions.
\end{example}

\section{Proof for the one period setting}
\label{sec:one.period}

Let $(\Omega,\mathcal{F})$ be a measurable space armed with a family of probability 
measures $\mathcal{P}\subset\mathfrak{P}(\Omega)$.
Further let $S_0\in\mathbb{R}^d$ and $S_1\colon\Omega\to\mathbb{R}^d$ be measurable and write $\Delta S :=S_1-S_0$.
We write $h\in\Theta=\mathbb{R}^d$ for trading strategies and assume the no-arbitrage
NA$(\mathcal{P})$, i.e.~$h\Delta S\geq 0$ $\mathcal{P}$-q.s.~implies 
$h\Delta S= 0$ $\mathcal{P}$-q.s.~for every $h\in\mathbb{R}^d$.
Given some random variable $Z\colon\Omega\to[0,+\infty)$, denote by
\[\mathcal{M}(Z)
=\{ Q\in\mathfrak{P}(\Omega) : E_Q[|\Delta S|+Z]+H(Q,\mathcal{P})<+\infty 
\text{ and } E_Q[\Delta S]=0\}\]
the set of martingale measures that have finite entropy and integrate $Z$.
The following is the main result of this section.

\begin{theorem}
\label{thm:1peroiod}
	Fix a random variable $X\colon\Omega\to(-\infty,+\infty]$.
	Then one has
	\begin{align}
	\label{eq:problem.1period}	
	\inf_{h\in\mathbb{R}^d}\sup_{P\in\mathcal{P}} \log E_P\big[\exp(X+h\Delta S)\big]
	 =\sup_{Q\in\mathcal{M}(Y)} \big( E_Q[X]-H(Q,\mathcal{P}) \big)
	\end{align}
	for every random variable 	$Y\colon\Omega\to[0,+\infty)$
	and both terms are not equal to $-\infty$.
	Moreover, the infimum over $h\in\mathbb{R}^d$ is attained.
\end{theorem}

The following lemma, which turns out to be be helpful in the multiperiod case,
is shown in the course of the proof of Theorem \ref{thm:1peroiod}. 

\begin{lemma}
\label{lem:1period.cont.below}
	Let $X_n\colon\Omega\to(-\infty,+\infty]$ be a sequence of random variables
	increasing point-wise to $X$. Then it holds
	\[\sup_n \inf_{h\in\mathbb{R}^d}\sup_{P\in\mathcal{P}} \log E_P\big[\exp(X_n+h\Delta S)\big]
	=\inf_{h\in\mathbb{R}^d}\sup_{P\in\mathcal{P}} \log E_P\big[\exp(X+h\Delta S)\big],\]
	i.e.~the optimization problem is continuous from below.
\end{lemma}

\begin{lemma}
\label{lem:rep.exponential}
	Fix a random variable $X\colon\Omega\to\mathbb{R}$.
	Then one has 
	\[\sup_{P\in\mathcal{P}}\log E_P[\exp(X+h\Delta S)]
	=\sup_{Q\in\mathcal{C}}\big(E_Q[X+h\Delta S] - H(Q,\mathcal{P}) \big)\]
	for every $h\in\mathbb{R}^d$, where
	\[ \mathcal{C}:=\{Q\in\mathfrak{P}(\Omega) : 
	E_Q[|X|+|\Delta S|+Y]+H(Q,\mathcal{P})<+\infty\} \]
	and $Y\colon\Omega\to[0,+\infty)$ is an arbitrary random variable.
\end{lemma}
\begin{proof}	
	(a) 
	Define $Z:=X+h\Delta S$ and fix a measure $P\in\mathcal{P}$.
	It follows from the well known representation of expected exponential 
	utility and the monotone convergence theorem that
	\begin{align}
	\label{eq:rep.exponential.set.smaller} 
	\log E_P[\exp(Z)] 
	=\sup_{Q\in\mathcal{A}_P}(E_Q[Z]-H(Q,P)),
	\end{align}
	where 
	\[\mathcal{A}_P:=\{Q\in\mathfrak{P}(\Omega) : E_Q[Z^-] + H(Q,P)<+\infty\}.\]
	For the sake of completeness, a proof is provided in Lemma \ref{lem:rep.exponential.dominated}.
	We claim that one can replace $\mathcal{A}_P$ with $\mathcal{C}_P$ 
	in \eqref{eq:rep.exponential.set.smaller} 
	without changing the value of the supremum, where
	\[ \mathcal{C}_P:=\{ Q\in\mathfrak{P}(\Omega): E_Q[|X|+|\Delta S|+Y]+H(Q,P)<+\infty\}. \]
	Since $\mathcal{C}_P$ is a subset of $\mathcal{A}_P$, it suffices
	to show that for any $Q\in\mathcal{A}_P$, there exists a sequence
	$Q_n\in\mathcal{C}_P$ such that $E_{Q_n}[Z]-H(Q_n,P)$ converges to $E_Q[Z]-H(Q,P)$.
	To that end, fix some $Q\in\mathcal{A}_P$ and define
	\[ Q_n:=Q(\,\cdot\,|B_n) \quad\text{where}\quad
	B_n:=\{ |X|+|\Delta S| + Y\leq n \}\]
	for all $n$ large enough such that $Q(B_n)>0$.
	Then it holds
	\[ \frac{dQ_n}{dP}=\frac{1_{B_n}}{Q(B_n)}\frac{dQ}{dP}\]
	and since $B_n\uparrow\Omega$, a straightforward computation shows that
	\[H(Q_n,P)
	=E_P\Big[\frac{1_{B_n}}{Q(B_n)}\frac{dQ}{dP}\log\frac{dQ}{dP}\Big] -\log Q(B_n)
	\to H(Q,P).\]
	In particular, $H(Q_n,P)<+\infty$ and since
	$X$, $\Delta S$, and $Y$ are integrable with respect to 
	$Q_n$, it follows that $Q_n\in\mathcal{C}_P$.
	Further, the integrability of $Z^-$ with respect to $Q$ guarantees 
	the convergence of $E_{Q_n}[Z]$ to $E_Q[Z]$ and therefore
	\[E_Q[Z]-H(Q,P)
	=\lim_n (E_{Q_n}[Z]-H(Q_n,P))
	\leq\sup_{Q\in\mathcal{C}_P}(E_Q[Z]-H(Q,P)).\]
	Taking the supremum over all $Q\in\mathcal{A}_P$ yields the claim.

	(b) 
	To conclude the proof, make the simple observation that
	$\mathcal{C}$ equals the union over $\mathcal{C}_P$ where $P$ runs trough $\mathcal{P}$. 
	This implies that
	\[\sup_{P\in\mathcal{P}} \log E_P[\exp(Z)]
	=\sup_{P\in\mathcal{P}} \sup_{Q\in\mathcal{C}_P}(E_Q[Z]-H(Q,P))
	=\sup_{Q\in\mathcal{C}}(E_Q[Z]-H(Q,\mathcal{P})),\]
	where the first equality follows from step (a).
\end{proof}

\begin{lemma}
\label{lem:convex.H.C}
	The relative entropy $H$ is jointly convex.
	Moreover, the function $H(\cdot,\mathcal{P})$
	and the set $\mathcal{C}$ defined in Lemma \ref{lem:rep.exponential} are convex.
\end{lemma}
\begin{proof}
	It follows from \cite[Lemma 3.29]{follmer2011stochastic} that
	\[ H(Q,P)=\sup\{ E_Q[Z] -\log E_P[\exp(Z)] : Z\text{ is a bounded random variable} \}. \]
	For any such $Z$, the function $(Q,P)\mapsto E_Q[Z] -\log E_P[\exp(Z)]$ is convex.
	Thus $H$, as the supremum over convex functions, is itself convex. 
	Furthermore, the convexity of $\mathcal{P}$ yields that
	$H(\cdot,\mathcal{P})$ and $\mathcal{C}$ are convex.
\end{proof}

In the proof of Theorem \ref{thm:1peroiod}, it will be important that 
$0\in\mathop{\mathrm{ri}}\{E_Q[\Delta S] :Q\in\mathcal{C}\}$
where $\mathcal{C}$ was defined in Lemma \ref{lem:rep.exponential} and $\mathop{\mathrm{ri}}$ denotes
the relative interior. 
To get the idea why this is true, assume for simplicity that $d=1$ and that $\Delta S$ 
is not $\mathcal{P}$-quasi surely equal to 0.
Then, by the no-arbitrage condition, there exist two measures 
$P^\pm$ such that $P^\pm(\pm\Delta S>0)>0$.
Now define
\[Q_\lambda:=\lambda P^+(\,\cdot\,|\,0< \Delta S,|X|,Y <n)+(1-\lambda) P^-(\,\cdot\,|-n<\Delta S,-|X|,-Y<0)\]
for $n$ large enough and every $\lambda\in[0,1]$.
Then $X$, $\Delta S$, and $Y$ are integrable with respect to $Q_\lambda$ and since
$E_{Q_0}[\Delta S]<0$, $E_{Q_1}[\Delta S]>0$ it follows that
$0\in\mathop{\mathrm{int}} \{ E_{Q_\lambda}[\Delta S] : \lambda\in [0,1]\}$.
As the density of $Q_\lambda$ with respect to $(P^++P^-)/2\in\mathcal{P}$ is bounded,
it holds $H(Q_\lambda,\mathcal{P})<+\infty$ and thus $Q_\lambda\in\mathcal{C}$.	

\begin{lemma}[\text{\cite[Lemma 3.3]{bouchard2015arbitrage}}]
	\label{lem:fundamental.lem}
	Let $X,Y\colon\Omega\to\mathbb{R}$ be random variables and assume that $Y$ is nonnegative.
	Then one has $0\in\mathop{\mathrm{ri}} \{ E_Q[\Delta S] : Q\in\mathcal{C}\}$
	where $\mathcal{C}$ was defined in Lemma \ref{lem:rep.exponential}.
\end{lemma}
\begin{proof}
	Even though \cite[Lemma 3.3]{bouchard2015arbitrage} states that
	$0\in\mathop{\mathrm{ri}} \{ E_Q[\Delta S] : Q\in\Theta\}$
	for the set
	$\Theta=\{Q: E_Q[|X|+|\Delta S|+Y]<+\infty\text{ and } Q\ll P\text{ for some }P\in\mathcal{P}\}$,
	the constructed measures $Q$ have bounded densities $dQ/dP$ with respect to some $P\in\mathcal{P}$, 
	in particular $H(Q,\mathcal{P})$ is finite.
	The proof can be copied word by word.
\end{proof}

Before being ready for the proof of the main theorem, one last observation on the 
decomposition of $\mathbb{R}^d$ into relevant and irrelevant strategies $h$ needs to be made.
Denote by $\mathop{\mathrm{supp}}_\mathcal{P}\Delta S$ the smallest closed subset of $\mathbb{R}^d$ 
such that $\Delta S(\omega)\in\mathop{\mathrm{supp}}_\mathcal{P}\Delta S$ for $\mathcal{P}$-quasi every $\omega$;
see \cite[Lemma 4.2]{bouchard2015arbitrage}.
Further write $\mathop{\mathrm{lin}}A$ for the smallest linear space which contains 
a given set $A\subset\mathbb{R}^d$, 
and $L^\perp:=\{h\in\mathbb{R}^d : hl=0\text{ for all }l\in L\}$ for the orthogonal 
complement of a linear space $L\subset\mathbb{R}^d$.

\begin{lemma}[\text{\cite[Lemma 2.6]{nutz2014utility}}]
\label{lem:char.space.L}
	Define $L:=\mathop{\mathrm{lin}}\mathop{\mathrm{supp}}_\mathcal{P}\Delta S$.
	Then one has $h\in L^\perp$ if and only if $h\Delta S=0$ $\mathcal{P}$-quasi surely.
\end{lemma}

\begin{proof}[{\bf Proof of Theorem \ref{thm:1peroiod} and Lemma \ref{lem:1period.cont.below}}]
	In step (a), duality is shown under the assumption that $X$ is bounded from above. 
	The existence of an optimizer $h\in\mathbb{R}^d$
	as well as continuity from below are proven simultaneously in step (b).
	Finally, the results from (a) and (b) are combined to extend 
	to unbounded random endowment $X$ in step (c).
	
	(a)
	Throughout this step, assume that $X$ is bounded from above,
	meaning that there exists some constant $k$ such that $X(\omega)\leq k$ for every $\omega$.
	The goal is to show the following dual representation
	\begin{align}
	\label{eq:1peroid.dual.Xintegr}		
		\inf_{h\in\mathbb{R}^d}\sup_{P\in\mathcal{P}} \log E_P[\exp(X+h\Delta S)]
	=\sup_{Q\in\mathcal{M}(|X|+Y)} (E_Q[X]-H(Q,\mathcal{P})).
	\end{align}
	By Lemma \ref{lem:rep.exponential}, it holds
	\begin{align*}
	\inf_{h\in\mathbb{R}^d}\sup_{P\in\mathcal{P}} \log E_P[\exp(X+h\Delta S)] 
	=\inf_{h\in\mathbb{R}^d}\sup_{Q\in\mathcal{C}}\big( E_Q[X+h\Delta S]-H(Q,\mathcal{P}) \big)
	\end{align*}
	where
	\[ \mathcal{C}:=\{Q\in\mathfrak{P}(\Omega) : E_Q[|X|+|\Delta S|+Y]+H(Q,\mathcal{P})<+\infty \}.\]
	Thus, if interchanging the infimum over $h\in\mathbb{R}^d$ 
	and the supremum over $Q\in\mathcal{C}$ were possible,  
	\eqref{eq:1peroid.dual.Xintegr} would follow
	since $\inf_{h\in\mathbb{R}^d}E_Q[h\Delta S]=-\infty$ whenever $Q$ is not a martingale measure.
	In what follows, we argue why one can in fact interchange the infimum and the supremum.
	Define
	\[\Gamma:=\mathop{\mathrm{lin}}\{ E_Q[\Delta S] :Q\in\mathcal{C} \}\]	
	and notice that if $\Gamma=\{0\}$, then $\mathcal{C}=\mathcal{M}(|X|+Y)$ 
	and $E_Q[h\Delta S]=0$ for all $h\in\mathbb{R}^d$ and $Q\in\mathcal{C}$ so that there is nothing to prove.
	Therefore, assume in the sequel that $\Gamma\neq\{0\}$ and let
	\[\{e_1,\dots,e_r\} \text{ be an orthonormal basis of }  \Gamma.\]
	Further, to simplify notation, define the function $J\colon \mathcal{C}\times\mathbb{R}^d\to\mathbb{R}$,
	\[ J(Q,h):=hE_Q[\Delta S]+ E_Q[X]-H(Q,\mathcal{P}). \]
	By Lemma \ref{lem:convex.H.C}, the set $\mathcal{C}$ and the function
	$H(\cdot,\mathcal{P})$ are convex, which shows that $J(\cdot,h)$ is concave for all $h\in\mathbb{R}^d$.
	Further, $J(Q,\cdot)$ is convex for all $Q\in\mathcal{C}$.
	Therefore, \cite[Theorem 4.1]{sion1958general} gives a sufficient condition for
	\[	\inf_{h\in\mathbb{R}^d}\sup_{Q\in\mathcal{C}} J(Q,h)
	=\sup_{Q\in\mathcal{C}}\inf_{h\in\mathbb{R}^d} J(Q,h) \]
	to hold true, namely that
	\begin{align}
	\label{eq:supinfcondition}
	\begin{cases}
	\text{for every } 
	c<\inf_{h\in\mathbb{R}^d}\sup_{Q\in\mathcal{C}} J(Q,h)
	\text{ one can find a finite set  }F\subset \mathcal{C}\\
	\text{such that for every }
	h\in\mathbb{R}^d \text{ there exists } Q\in F \text{ satisfying }  J(Q,h)>c.
	\end{cases}
	\end{align}
	To prove \eqref{eq:supinfcondition}, fix such $c$ and notice that
	\[\{h\in\mathbb{R}^d: J(Q,h)>c\}
	=\{h\in \Gamma : J(Q,h)>c\}+ \Gamma^\perp\]
	since $hE_Q[\Delta S]=0$ for every $h\in\Gamma^\perp$ and $Q\in\mathcal{C}$.
	Therefore, we can assume without loss of generality that $h\in\Gamma$ in the sequel.
	In fact, we shall distinguish between elements in $\Gamma$ with large and small (Euclidean) length.
	From Lemma \ref{lem:fundamental.lem}, it follows that
	\[ 0\in\mathop{\mathrm{ri}}\{ E_Q[\Delta S] : Q\in\mathcal{C}\} \]
	which implies that there exist $a^\pm_i>0$ and $Q_i^\pm\in\mathcal{C}$ satisfying
	\[E_{Q_i^\pm}[\Delta S]=\pm a^\pm_ie_i \quad\text{for } 1\leq i\leq r.\]
	We claim that  
	\begin{align}
	\label{eq:for.existence.minimizer}
	 \begin{cases}
	\max\{J(h,Q_i^\pm) : 1\leq i\leq r\}>c+1>c\\
	\text{for all } h\in\Gamma\text{ such that }
	|h|>m\sqrt{r}/\delta
	\end{cases}
	\end{align}
	where
	\[ m:=\max\big\{c+1 - E_{Q^\pm_i}[X] + H(Q^\pm_i,\mathcal{P}) : 1\leq i\leq r\big\}\in\mathbb{R}\]
	and
	\[ \delta:=\min\{ a^\pm_i : 1\leq i\leq r \}>0.\]
	Indeed, since $\sum_{i=1}^r (he_i)^2=|h|^2> r(m/\delta)^2$, it follows that
	$|he_j|> m/\delta$ for some $1\leq j\leq r$. If $he_j>m/\delta$, it holds
	\[hE_{Q^+_j}[\Delta S]
	=h a^+_je_j
	>\frac{m a_j^+}{\delta}
	\geq  m
	\geq c+1- E_{Q^+_j}[X] +H(Q^+_j,\mathcal{P})\]
	and a rearrangement of the appearing terms yields $J(h,Q_j^+)>c+1$.
	If $he_j<-m/\delta$, the same argumentation shows that $J(h,Q_j^-)>c+1$.
	Further, as 
	\[ J(Q,\cdot) \text{ is continuous and } c<\inf_{h\in\Gamma}\sup_{Q\in\mathcal{C}} J(Q,h),\]
	the collection
	\[U_Q :=\{h \in \Gamma:  J(Q,h) >c \},\]
	where $Q\in\mathcal{C}$, forms an open cover of $\Gamma$.
	By compactness of the set $\{h\in \Gamma : |h|\leq m\sqrt{r}/\delta\}$, 
	there exists a finite family $F'\subset\mathcal{C}$ such that 
	\[\{h\in \Gamma : |h|\leq m\sqrt{r}/\delta\}\subset \bigcup \{U_Q : Q\in F'\}.\]
	Then $F:=F'\cup\{Q^\pm_i : 1\leq i\leq r\}$ is still finite and it holds
	\[ \Gamma=\bigcup \{ U_Q : Q\in F\},\]
	which is a reformulation of \eqref{eq:supinfcondition}.
	Putting everything together, it follows that
	\begin{align*}
	&\inf_{h\in\mathbb{R}^d}\sup_{P\in\mathcal{P}} \log E_P[\exp(X+h\Delta S)]
	=\inf_{h\in\mathbb{R}^d}\sup_{Q\in\mathcal{C}} J(Q,h)\\
	&=\sup_{Q\in\mathcal{C}}\inf_{h\in\mathbb{R}^d} J(Q,h)
	=\sup_{Q\in\mathcal{M}(|X|+Y)} (E_Q[X]-H(Q,\mathcal{P})).
	\end{align*}
	In particular, since $\mathcal{M}(|X|+Y)$ is not empty by Lemma \ref{lem:fundamental.lem},
	it follows that the optimization problem does not take the value $-\infty$.

	(b) 
	We proceed to show that the optimization problem is continuous from below 
	(Lemma \ref{lem:1period.cont.below})
	and that an optimal strategy $h\in\mathbb{R}^d$ exists. 
	Recall that $X_n$ is a sequence increasing point-wise to $X$.
	For the existence of an optimal strategy for a fixed function $X$, consider
	the constant sequence $X_n:=X$ in the following argumentation.
	For each natural number $n$, let $h_n\in\mathbb{R}^d$ such that
	\begin{align}
	\label{eq:h.nearly.optimal}
	\inf_{h\in\mathbb{R}^d} \sup_{P\in\mathcal{P}} \log E_P[\exp(X_n+h\Delta S)]
	\geq \sup_{P\in\mathcal{P}} \log E_P[\exp(X_n+h_n\Delta S)]- \frac{1}{n}.
	\end{align}
	By step (a) this is possible, i.e.~the left-hand side of \eqref{eq:h.nearly.optimal}
	is not equal to $-\infty$.
	By Lemma \ref{lem:char.space.L} we may assume without loss of generality that 
	every $h_n$ is an element of 
	$L:=\mathop{\mathrm{lin}}\mathop{\mathrm{supp}}\nolimits_\mathcal{P}\Delta S$.
	
	First, assume that the sequence $h_n$ is unbounded, 
	i.e.~$\sup_n|h_n|=+\infty$. 
	Then, possibly after passing to a subsequence, $h_n/|h_n|$ converges to some limit $h^\ast$. 
	Since $|h^\ast|=1$ and $h^\ast\in L$,
	it follows from Lemma \ref{lem:char.space.L} and the NA($\mathcal{P}$)-condition, that 
	$P'(A)>0$ for some $P'\in\mathcal{P}$ where $A:=\{h^\ast\Delta S>0\}$. 
	However, since 
	\[\exp(X_n+h_n\Delta S)1_A\to+\infty 1_A,\]
	an application of Fatou's lemma yields
	\[ \inf_{h\in\mathbb{R}^d}\sup_{P\in\mathcal{P}} \log E_P[\exp(X_n+h\Delta S)]
	\geq \log E_{P'}[\exp(X_n+h_n\Delta S)]-\frac{1}{n}
	\to +\infty.\]
	But then, since the sequence $X_n$ is increasing, it follows that
	\begin{align*}
	&\inf_{h\in\mathbb{R}^d}\sup_{P\in\mathcal{P}} \log E_P[\exp(X+h\Delta S)]\\
	&\geq \lim_n \inf_{h\in\mathbb{R}^d}\sup_{P\in\mathcal{P}} \log E_P[\exp(X_n+h\Delta S)]
	=+\infty.
	\end{align*}
	Hence the optimization problem is continuous from below 
	and every $h\in\mathbb{R}^d$ is optimal for $X$.

	If the sequence $h_n$ is bounded, again possibly after passing to a subsequence, 
	$h_n$ converges to some limit $h^\ast\in\mathbb{R}^d$. 
	Now it follows that
	\begin{align*}
	\sup_{P\in\mathcal{P}} \log E_P[\exp(X+h^\ast\Delta S)]
	&\leq \liminf_n\Big( \sup_{P\in\mathcal{P}} \log E_P[\exp(X_n+h_n\Delta S)]-\frac{1}{n}\Big)\\
	&\leq \liminf_n	\inf_{h\in\mathbb{R}^d} \sup_{P\in\mathcal{P}} \log E_P[\exp(X_n+h\Delta S)] \\	
	&\leq \inf_{h\in\mathbb{R}^d} \sup_{P\in\mathcal{P}} \log E_P[\exp(X+h\Delta S)],
	\end{align*}
	where the first inequality follows from Fatou's lemma, the second one  
	since $h_n$ was chosen optimal up to an error of $1/n$, and the last one 
	since $X_n$ is an increasing sequence.
	This shows both that the optimization problem is continuous from below 
	and that $h^\ast$ is optimal for $X$.

	(c) 
	In the final step, the duality established in (a) is extended to general
	random endowment. Let $X\colon\Omega\to(-\infty,+\infty]$ be measurable and observe that
	\[\mathcal{M}(X^- + Y)=\mathcal{M}(|X\wedge n| + Y)\quad\text{for all }n\in\mathbb{N}\]
	since $X^-$ is integrable if and only if $(X\wedge n)^-$ is.
	Moreover, for any $Q\in\mathcal{M}(X^- + Y)$ the monotone convergence theorem applies
	and yields $\sup_n E_Q[X\wedge n]=E_Q[X]$. 
	But then it follows that
	\begin{align*}
	&\inf_{h\in\mathbb{R}^d}\sup_{P\in\mathcal{P}} \log E_P[\exp(X + h\Delta S)]
	=\sup_n \inf_{h\in\mathbb{R}^d}\sup_{P\in\mathcal{P}} \log E_P[\exp(X\wedge n+ h\Delta S)]\\ 
	&=\sup_n \sup_{Q\in\mathcal{M}(X^- + Y)} \big( E_Q[X\wedge n]-H(Q,\mathcal{P})\big)
	=\sup_{Q\in\mathcal{M}(X^- + Y)} \big( E_Q[X]-H(Q,\mathcal{P})\big)
	\end{align*}	 	
	where the first and second equality follow from step (b) and (a), respectively, and the last
	one by interchanging two suprema.
\end{proof}

\section{Proofs for the multiperiod case}
\label{sec:multi.period}

\subsection{The case without options}

In this section, measurable selection arguments are used to show that the global analysis 
can be reduced to a local one wherein the results of the one period case are used.
For each $0\leq t \leq T-1$ and $\omega\in\Omega_t$, define
\[\mathcal{P}_t^T(\omega)=\{P_t\otimes\cdots\otimes P_{T-1}: 
P_s(\cdot)\in\mathcal{P}_s(\omega,\cdot) \text{ for }t\leq s\leq T-1\},\]
where each $P_s(\cdot)$ is a universally measurable selector of $\mathcal{P}_s(\omega,\cdot)$.
Thus $\mathcal{P}_t^T(\omega)$ corresponds to the set of all possible probability scenarios 
for the future stock prices $S_{t+1},\dots, S_T$, given the past $\omega\in\Omega_t$.
In particular, it holds $\mathcal{P}=\mathcal{P}_0^T$ in line with Bouchard and Nutz.
In order to keep the indices to a minimum, fix two functions
\[ X\colon\Omega\to(-\infty,+\infty]\quad\text{and}\quad Y\colon\Omega\to[0,+\infty)\]
such that $X$ and $-Y$ are  upper semianalytic, and define the set of all martingale measures 
for the future stock prices $S_{t+1},\dots, S_T$ given the past $\omega\in\Omega_t$ by
\[ \mathcal{M}_t^T(\omega)
=\bigg\{ Q\in\mathfrak{P}(\Omega_{T-t}): 
\begin{array}{l} 
(S_s(\omega,\cdot))_{t\leq s \leq T} \text{ is a $Q$-martingale and}\\
E_Q[X(\omega,\cdot)^-+Y(\omega,\cdot)] + H(Q,\mathcal{P}_t^T(\omega))<+\infty
\end{array}\bigg\}\]
for each $0\leq t\leq T-1$ and $\omega\in\Omega_t$.
It is shown in Lemma \ref{lem:graph.Q.analytic} that $\mathcal{M}_t^T$ has analytic graph
and within the proof of Theorem \ref{thm:multiperiod} that its values are not empty.
Note that $\mathcal{M}_0^T=\mathcal{M}(Y)=\{Q\in\mathcal{M} : E_Q[Y+X^-]<+\infty \}$,
where $\mathcal{M}$ was defined in Section \ref{sec:main}.
Further introduce the dynamic version of the optimization problem:
Define 
\[\mathcal{E}_T(\omega,x):=X(\omega)+x \]
for $(\omega,x)\in\Omega\times\mathbb{R}$ 
and recursively
\begin{align}
\label{eq:local.optimization}
\mathcal{E}_t(\omega,x)
:=\inf_{h\in\mathbb{R}^d}\sup_{P\in\mathcal{P}_t(\omega)}\log 
E_P\Big[\exp\Big(\mathcal{E}_{t+1}\big(\omega\otimes_t\cdot,x+h\Delta S_{t+1}(\omega,\cdot)\big)\Big)\Big]
\end{align}
for $(\omega,x)\in\Omega_t\times\mathbb{R}$.
Here we write $\omega\otimes_t\omega':=(\omega,\omega')\in\Omega_{t+s}$ 
for $\omega\in\Omega_t$ and $\omega'\in\Omega_s$ instead of $(\omega,\cdot)$ to avoid confusion.
It will be shown later that $\mathcal{E}_t$ is well defined, 
i.e.~that the term inside the expectation is appropriately measurable.

The following theorem is the main result of this section and includes 
Theorem \ref{thm:main} as a special case (corresponding to $t=0$).

\begin{theorem}
\label{thm:multiperiod}
	For every $0\leq t\leq T-1$ and $\omega\in\Omega_t$, it holds
	\begin{align*}
	\mathcal{E}_t(\omega,x)-x
	&=\inf_{\vartheta\in\Theta}\sup_{P\in\mathcal{P}_t^T(\omega)}\log
	E_P\Big[\exp\Big(X(\omega,\cdot) + (\vartheta\cdot S)_t^T(\omega,\cdot)\Big)\Big] \\
	&=\sup_{Q\in\mathcal{M}_t^T(\omega)}
	\Big(E_Q[X(\omega,\cdot)] -H(Q,\mathcal{P}_t^T(\omega)) \Big)
	\end{align*}
	and both terms are not equal to $-\infty$.
	Moreover, the infimum over 	$\vartheta\in\Theta$ is attained.
\end{theorem}

We start by investigating properties of the (robust) relative entropy 
and the graph of $\mathcal{M}_t$, which will ensure that measurable selection arguments
can be applied.
We then focus on deriving a duality for $\mathcal{E}_t$ and last prove the dynamic
programming principle.

\begin{lemma}[\text{\cite[Lemma 1.4.3.b]{dupuis2011weak}}]
\label{lem:H.is.borel}
	The relative entropy $H$ is Borel.
\end{lemma}
\begin{proof}
	Any Borel function can be approximated in measure by continuous functions,
	so it follows as in the proof of Lemma \ref{lem:convex.H.C} that
	\[ H(Q,P)=\sup\{ E_Q[Z] -\log E_P[\exp(Z)] : Z\text{ is bounded and continuous} \}. \]
	Therefore $\{H\leq c\}$ is closed for any real number $c$ showing that $H$ is Borel.
\end{proof}

The so-called chain rule for the relative entropy is well known, a proof can be found
e.g.~in Appendix C3 of the book by Dupuis and Ellis \cite{dupuis2011weak}.
However, since we are dealing with universally measurable kernels  
and also in order to be self-contained, a proof is given in the Appendix.
For the link between dynamic risk measures and this chain rule see 
e.g.~\cite{cheridito2011composition} in the dominated, 
and \cite{bartl2016conditional,lacker2015liquidity} in the nondominated setting.

\begin{lemma}
\label{lem:entropie.sum}
	Let $0\leq t\leq T-1$ and $P,Q\in\mathfrak{P}(\Omega_{T-t})$. Then 
	\[ H(Q,P)=\sum_{s=t}^{T-1} E_Q[H(Q_s(\cdot),P_s(\cdot))] \]	
	where $Q_s$ and $P_s$ are universally measurable kernels such that 
	$Q=Q_t\otimes\cdots\otimes Q_{T-1}$ and $P=P_t\otimes\cdots\otimes P_{T-1}$.	
\end{lemma}

\begin{lemma}
\label{lem:entropie.sum.robust}
	For any $0\leq t\leq T-1$, the function
	\[ \Omega_t\times\mathfrak{P}(\Omega_{T-t})\to[-\infty,0],\quad
	(\omega,Q)\mapsto -H(Q,\mathcal{P}_t^T(\omega))\]
	is upper semianalytic. Moreover, it holds
	\begin{align*} 
	H(Q,\mathcal{P}_t^T(\omega))
	&=H(Q_t,\mathcal{P}_t(\omega)) +E_Q[ H(Q'(\cdot),\mathcal{P}_{t+1}^T(\omega,\cdot))]
	\end{align*}
	where $Q'$ is a universally measurable kernel such that $Q=Q_t\otimes Q'$.
\end{lemma}
\begin{proof}	
	Every probability $Q\in\mathfrak{P}(\Omega_{T-t})$ can be written as 
	$Q=Q_t\otimes\cdots\otimes Q_{T-1}$ where $Q_s$ are the kernels
	from Remark \ref{rem:kernels.measurable.and.derivative.measurable}, i.e.~such that
	\[ \Omega_{s-t}\times\mathfrak{P}(\Omega_{T-t})\to\mathfrak{P}(\Omega_1),
	\quad(\bar{\omega},Q)\mapsto Q_s(\bar{\omega})\] 
	is Borel.

	(a) We start by showing that
	\begin{align}
	\label{eq:H.equal.sum}
	\Omega_t\times\mathfrak{P}(\Omega_{T-t})\to[-\infty,0],
	\quad (\omega,Q)\mapsto \sum_{s=t}^{T-1} -E_Q[ H(Q_s(\cdot),\mathcal{P}_s(\omega,\cdot))]
	\end{align}
	is upper semianalytic. Fix some $t\leq s\leq T-1$. 
	In the sequel, $\omega$ will refer to elements in $\Omega_t$
	and $\bar{\omega}$ to elements in $\Omega_{s-t}$.
	Since $(\bar{\omega},Q)\mapsto Q_s(\bar{\omega})$ is Borel by construction
	and the entropy $H$ is Borel by Lemma \ref{lem:H.is.borel}, the composition
	\begin{align*}
	\Omega_t\times\Omega_{s-t}\times\mathfrak{P}(\Omega_{T-t})\times\mathfrak{P}(\Omega_1)
	&\to[-\infty,0],\quad
	(\omega,\bar{\omega},Q,R)
	\mapsto -H(Q_s(\bar{\omega}),R)
	\end{align*}
	is Borel as well.
	As the graph of $\mathcal{P}_s$ is analytic,
	it follows from \cite[Proposition 7.47]{bertsekas1978stochastic} that
	\begin{align}
	\label{eq:H(Q,scrP).is.lsa}
	\Omega_t\times\Omega_{s-t}\times\mathfrak{P}(\Omega_{T-t})
	\to[-\infty,0],\quad
	&(\omega,\bar{\omega},Q)
	\mapsto 
	-H(Q_s(\bar{\omega}),\mathcal{P}_s(\omega,\bar{\omega}))
	\end{align}
	is upper semianalytic.
	Moreover, \cite[Proposition 7.50]{bertsekas1978stochastic} guarantees that for any $\varepsilon>0$,
	there exists a universally measurable kernel $P^\varepsilon_s$ such that
	\begin{align}
	\label{eq:selector.Ps'}
	\begin{cases}
	P^\varepsilon_s(\omega,\bar{\omega},Q)\in\mathcal{P}_s(\omega,\bar{\omega}), \\
	H(Q_s(\bar{\omega}),P^\varepsilon_s(\omega,\bar{\omega},Q))
	\leq H(Q_s(\bar{\omega}),\mathcal{P}_s(\omega,\bar{\omega}))+\varepsilon
	\end{cases}
	\end{align}
	for all $(\omega,\bar{\omega},Q)$. This will be used in part (b).
	Further, since
	\begin{align*}
	\Omega_t\times\mathfrak{P}(\Omega_{T-t})
	\to[-\infty,0],\quad
	(\omega,Q)
	\mapsto 
	-E_Q[H(Q_s(\cdot),\mathcal{P}_s(\omega,\cdot))]
	\end{align*}
	is just \eqref{eq:H(Q,scrP).is.lsa} integrated with respect to $Q(d\bar{\omega})$,
	an application of Lemma \ref{lem:integral.is.measurable} shows that this mapping
	is upper semianalytic.
	Finally, the fact that sums of upper semianalytic functions are again upper semianalytic 
	(see \cite[Lemma 7.30]{bertsekas1978stochastic}) implies that
	\eqref{eq:H.equal.sum} is upper semianalytic as was claimed.

	(b)
	Fix some $\omega\in\Omega_t$ and $Q\in\mathfrak{P}(\Omega_{T-t})$.
	From Lemma \ref{lem:entropie.sum}, it follows that
	\[ H(Q,\mathcal{P}_t^T(\omega))
	=\inf_{P\in\mathcal{P}_t^T(\omega)}\sum_{s=t}^{T-1} E_Q[ H(Q_s(\cdot),P_s(\cdot))]
	\geq \sum_{s=t}^{T-1} E_Q\big[ H(Q_s(\cdot),\mathcal{P}_s(\omega,\cdot))\big].\]
	For the other inequality, let $\varepsilon>0$ be arbitrary and 
	$P_s^\varepsilon$ be the kernels from \eqref{eq:selector.Ps'}.
	Recall that $Q$ and $\omega$ are fixed so that 
	\[P'_s\colon\Omega_{s-t}\to \mathfrak{P}(\Omega_1),\quad
	\bar{\omega}\mapsto P_s^\varepsilon(\omega,\bar{\omega},Q)\] 
	is still universally measurable by \cite[Lemma 7.29]{bertsekas1978stochastic}.
	Then it follows that
	\[P':=P_t'\otimes\cdots\otimes P_{T-1}'\in\mathcal{P}_t^T(\omega)\]
	and, using  Lemma \ref{lem:entropie.sum} once more, that
	\begin{align*}
	&\sum_{s=t}^{T-1} E_Q\big[ H(Q_s(\cdot),\mathcal{P}_s(\omega,\cdot))\big]
	\geq \sum_{s=t}^{T-1} E_Q\big[ H(Q_s(\cdot),P_s'(\cdot)) - \varepsilon   \big]\\
	&= H(Q,P') -(T-t)\varepsilon 
	\geq H(Q,\mathcal{P}_t^T(\omega)) -(T-t)\varepsilon.
	\end{align*}
	As $\varepsilon$ was arbitrary, this shows the desired inequality.
	
	(c)	
	Finally, kernels are almost-surely unique so that
	\[Q'=Q_{t+1}\otimes\cdots\otimes Q_{T-1}\quad Q_t\text{-almost surely.}\]
	Hence it follows that
	\[ H(Q'(\cdot),\mathcal{P}_{t+1}(\omega,\cdot))
	=\sum_{s= t+1}^{T-1} E_{Q'(\cdot)}[H(Q_s(\cdot),\mathcal{P}_s(\omega,\cdot))]
	\quad Q_t\text{-almost surely}.\]
	It only remains to integrate this equation with respect to $Q_t$.
\end{proof}

Fix a measure $Q=Q_t\otimes\cdots\otimes Q_{T-1}\in\mathfrak{P}(\Omega_{T-t})$
and $\omega\in\Omega_t$.
An elementary computation shows that
$Q$ is a martingale measure for $(S_s(\omega,\cdot))_{t\leq s\leq T}$ if and only if
$E_Q[|\Delta S_{s+1}(\omega,\cdot)|]<+\infty$ and
\[ E_{Q_s(\bar{\omega})}[\Delta S_{s+1}(\omega,\bar{\omega},\cdot)]=0 \quad
\text{for }  Q_t\otimes\cdots\otimes Q_{s-1}\text{-almost every }\bar{\omega}\in\Omega_{s-t}\]
and every $t\leq s\leq T-1$. This is used in the sequel without reference.

\begin{lemma}
\label{lem:graph.Q.analytic}
	The graph of $\mathcal{M}_t^T$ is analytic.
\end{lemma}
\begin{proof}
	First, notice that $Z:=X\wedge 0-Y$ is upper semianalytic.
	This follows from the fact that $\{X\wedge 0 \geq a\}$ equals $\emptyset$ if $a>0$ 
	and $\{X\geq a\}$ else and that the sum of upper semianalytic functions remains upper semianalytic.
	Therefore, an application of Lemma \ref{lem:integral.is.measurable} shows that 	
	\[\Omega_t\times\mathfrak{P}(\Omega_{T-t})\to[-\infty,0],
	\quad(\omega,Q)\mapsto E_Q[Z(\omega,\cdot)] \] 
	is upper semianalytic. 
	Then, since	$(\omega,Q)\mapsto -H(Q,\mathcal{P}_t^T(\omega))$ is upper semianalytic
	by Lemma \ref{lem:entropie.sum.robust} and the sum of upper semianalytic
	mappings is again upper semianalytic,	it follows that
	\[ A:=\{(\omega,Q): E_Q[Z(\omega,\cdot)] -  H(Q,\mathcal{P}_t^T(\omega))>-\infty\}\]
	is an analytic set.	
	The missing part now is the martingale property.
	First, notice that
	\[\Omega_t\times\mathfrak{P}(\Omega_{T-t})\to[0,+\infty],\quad
	(\omega,Q)\mapsto E_Q[|\Delta S_{s+1}(\omega,\cdot)|]  \]
	is Borel by Lemma \ref{lem:integral.is.measurable}.
	As before, for every $Q\in\mathfrak{P}(\Omega_{T-t})$, 
	we will write $Q=Q_t\otimes\cdots\otimes Q_{T-1}$ for the kernels $Q_s$ 
	from Remark \ref{rem:kernels.measurable.and.derivative.measurable}.
	Then, since $(\bar{\omega},Q)\mapsto Q_s(\bar{\omega})$ is Borel,
	a twofold application of Lemma \ref{lem:integral.is.measurable} shows that
	\[\Omega_t\times\mathfrak{P}(\Omega_{T-t})\to[0,+\infty],\quad
	(\omega,Q)\mapsto E_Q[ |E_{Q_s(\cdot)}[\Delta S_{s+1}(\omega,\cdot)]|]  \]
	is Borel. 
	Thus 
	\[ B_s:=\{(\omega,Q): E_Q[|\Delta S_{s+1}(\omega,\cdot)|] <+\infty \text{ and } 
	E_Q[ |E_{Q_s(\cdot)}[\Delta S_{s+1}(\omega,\cdot)]|]=0\}\]
	is Borel which implies that
	\[\mathop{\mathrm{graph}}\mathcal{M}_t^T
	= \bigcap \{ B_s : t\leq s\leq T-1\}\cap A, \]  
	as the finite intersection of analytic sets, is itself analytic 
	(see \cite[Corollary 7.35.2]{bertsekas1978stochastic}).
\end{proof}

Define
\[ \mathcal{D}_t(\omega):=\sup_{Q\in\mathcal{M}_t^T(\omega)}
\Big( E_Q[X(\omega,\cdot)]-H(Q,\mathcal{P}_t^T(\omega))\Big)\]  
for all $0\leq t\leq T-1$ and $\omega\in\Omega_t$, and recall that
\[\mathcal{E}_t(\omega,x)
:=\inf_{h\in\mathbb{R}^d}\sup_{P\in\mathcal{P}_t(\omega)}\log 
E_P\Big[\exp\big(\mathcal{E}_{t+1}\big((\omega,\cdot),x+h\Delta S_{t+1}(\omega,\cdot)\big)\big)\Big]\]
for $(\omega,x)\in\Omega_t\times\mathbb{R}$.

\begin{proof}[\text{\bf Proof of Theorem \ref{thm:multiperiod} -- Duality}]
	We claim that
	\begin{align}
	\label{eq:induction.main.thm}
	\bigg\{
	\begin{array}{l}
	\mathcal{E}_t(\omega,x)=\mathcal{D}_t(\omega)+x\quad\text{and}\quad \mathcal{D}_t(\omega)\in(-\infty,+\infty]\\
	\text{for all $\omega\in\Omega_t$, $x\in\mathbb{R}$ and $0\leq t\leq T-1$.}
	\end{array}
	\end{align}
	The proof will be a backward induction.
	For $t=T-1$, \eqref{eq:induction.main.thm} is just the statement of Theorem \ref{thm:1peroiod}.

	Now assume that \eqref{eq:induction.main.thm} holds true for $t+1$.
	First, we artificially bound $X$ from above and then pass to the limit.
	More precisely, define
	\[ \mathcal{D}^n_s(\omega):=\sup_{Q\in\mathcal{M}_s^T(\omega)} 
	\big( E_Q[X(\omega,\cdot)\wedge n] - H(Q,\mathcal{P}_s^T(\omega))\big)\]
	for $s=t,t+1$ and $\omega\in\Omega_s$, and notice that $\mathcal{D}^n_s$ is upper 
	semianalytic. Indeed, since $X(\omega,\cdot)\wedge n$ is upper semianalytic, the mapping
	$(\omega,Q)\mapsto E_Q[X(\omega,\cdot)\wedge n]$ 
	is upper semianalytic by Lemma \ref{lem:integral.is.measurable}.
	Then Lemma \ref{lem:entropie.sum.robust} and the fact that
	the sum of upper semianalytic functions stays upper semianalytic  implies that
	\[ (\omega,Q)\mapsto E_Q[X(\omega,\cdot)\wedge n] - H(Q,\mathcal{P}_s^T(\omega))\]
	is upper semianalytic.
	Since the graph of $\mathcal{M}_s^T$ is analytic by Lemma \ref{lem:graph.Q.analytic},
	it follows from \cite[Proposition 7.47]{bertsekas1978stochastic} that $\mathcal{D}_s^n$ is upper semianalytic.
	Moreover,  \cite[Proposition 7.50]{bertsekas1978stochastic} guarantees that
	for any $\varepsilon>0$ there exists a universally measurable kernel 
	$Q^\varepsilon(\cdot)\in\mathcal{M}_{t+1}^T(\omega,\cdot)$ such that
	\begin{align} 
	\label{eq:Q.eps.in.dual}
	\mathcal{D}^n_{t+1}(\omega\otimes_t\cdot)\leq 
	E_{Q^\varepsilon(\cdot)}[X(\omega,\cdot)\wedge n] 
	- H({Q^\varepsilon(\cdot)},\mathcal{P}_{t+1}^T(\omega,\cdot))+\varepsilon.
	\end{align}
	By interchanging two suprema, it holds $\mathcal{D}_s=\sup_n \mathcal{D}_s^n$
	(for more details, see part (c) of the proof of Theorem \ref{thm:1peroiod}).
	In particular, $\mathcal{D}_s$ is upper semianalytic, 
	as the countable supremum over upper semianalytic functions.
	Therefore, it follows from Lemma \ref{lem:1period.cont.below} that
	$\mathcal{E}_t=\sup_n \mathcal{E}_t^n$
	where
	\[\mathcal{E}_t^n(\omega,x)
	:=\inf_{h\in\mathbb{R}^d}\sup_{P\in\mathcal{P}_t(\omega)}
	\log E_P[\exp( \mathcal{D}^n_{t+1}(\omega\otimes_t\cdot) + h\Delta S_{t+1}(\omega,\cdot)+x)].\]
	The goal now is to show that 
	$\mathcal{E}_t^n$ equals $\mathcal{D}_t^n$ for all $n$, from which it follows that
	\[ \mathcal{E}_t(\omega,x)
	=\sup_n \mathcal{E}_t^n(\omega,x)
	=\sup_n \mathcal{D}_t^n(\omega)+x
	=\mathcal{D}_t(\omega)+x\]
	and the proof is complete.
	To show that indeed 
	$\mathcal{E}_t^n(\omega,x)=\mathcal{D}_t^n(\omega)+x$,
	fix some $n$, $x$, and $\omega\in\Omega_t$.
	By Theorem \ref{thm:1peroiod}, it holds
	\[ \mathcal{E}^n_t(\omega,x)
	=\sup_{Q_t\in\mathcal{M}_t(Z)}  \big( E_{Q_t}[\mathcal{D}^n_{t+1}(\omega\otimes_t\cdot)]
	-H(Q_t,\mathcal{P}_t(\omega))\big) + x
	>-\infty\]
	where
	\[\mathcal{M}_t(Z):=\bigg\{ Q\in\mathfrak{P}(\Omega_1) : 
	\begin{array}{l}
	E_Q[\mathcal{D}_{t+1}^n(\omega\otimes_t\cdot)^- +|\Delta S_{t+1}(\omega,\cdot)|+Z]<+\infty,\\
	H(Q,\mathcal{P}_t(\omega))<+\infty \text{ and } E_Q[\Delta S_{t+1}(\omega,\cdot)]=0
	\end{array} \bigg\}\]
	and $Z\colon\Omega_1\to[0,+\infty)$ is an arbitrary universally measurable function.

	We start by showing that $\mathcal{E}_t^n(\omega,x)\leq\mathcal{D}_t^n(\omega)+x$.
	Fix some $\varepsilon>0$, let 
	$Q^\varepsilon(\cdot)\in\mathcal{M}_{t+1}^T(\omega,\cdot)$ 
	be the kernel from \eqref{eq:Q.eps.in.dual}, and define $Z\colon\Omega_1\to[0,+\infty)$,
	\[Z:= E_{Q^{\varepsilon}(\cdot)}\Big[X(\omega,\cdot)^-+Y(\omega,\cdot)
	+\sum_{s= t+2}^T |\Delta S_s(\omega,\cdot)|\Big]
	+H(Q^\varepsilon(\cdot),\mathcal{P}_{t+1}^T(\omega,\cdot)).  \]
	Then $Z$ is real-valued by the definition of $\mathcal{M}_{t+1}^T(\omega,\cdot)$
	and it follows from Lemma \ref{lem:integral.is.measurable}, 
	Lemma \ref{lem:entropie.sum.robust}, and 
	\cite[Proposition 7.44]{bertsekas1978stochastic}
	that $Z$ is universally measurable.
	Moreover,
	\begin{align}
	\label{eq:Q.times.Qeps.in.Ms}
	Q_t\otimes Q^\varepsilon\in\mathcal{M}_t^T(\omega)\quad\text{for any }Q_t\in\mathcal{M}_t(Z).
	\end{align} 
	To show this, fix some $Q_t\in\mathcal{M}_t(Z)$ and define $Q:=Q_t\otimes Q^\varepsilon$.
	Then an application of Lemma \ref{lem:entropie.sum.robust} yields
	\begin{align*} 
	H(Q,\mathcal{P}_t^T(\omega))
	&=H(Q_t,\mathcal{P}_t(\omega))+ E_{Q_t}[H(Q^\varepsilon(\cdot),\mathcal{P}_{t+1}^T(\omega,\cdot))]\\
	&\leq H(Q_t,\mathcal{P}_t(\omega)) + E_{Q_t}[Z]
	<+\infty.
	\end{align*}
	Moreover, it holds
	\[E_Q\Big[X(\omega,\cdot)^- +Y(\omega,\cdot)+ \sum_{s= t+1}^T|\Delta S_s(\omega,\cdot)|\Big]
	\leq E_{Q_t}[ |\Delta S_{t+1}(\omega,\cdot)|+ Z]
	<+\infty\]
	so that indeed $Q\in\mathcal{M}_t^T(\omega)$ and, therefore, 
	\begin{align*} 
	& E_{Q_t}[\mathcal{D}_{t+1}^n(\omega\otimes_t\cdot)]- H(Q_t,\mathcal{P}_t(\omega))\\
	&\leq E_{Q_t}[ E_{Q^\varepsilon(\cdot)}[X(\omega,\cdot)\wedge n]
		- H(Q^\varepsilon(\cdot),\mathcal{P}_{t+1}^T(\omega,\cdot))+\varepsilon] 
		- H(Q_t,\mathcal{P}_t(\omega))\\
	&= E_Q[X(\omega,\cdot)\wedge n]- H(Q,\mathcal{P}_t^T(\omega))+\varepsilon
	\leq \mathcal{D}_t^n(\omega)+\varepsilon.
	\end{align*}
	As $Q_t\in\mathcal{M}_t(Z)$ and $\varepsilon>0$ were arbitrary, it follows that 
	$\mathcal{E}_t^n(\omega,x)\leq\mathcal{D}_t^n(\omega)+x$.

	To show the other inequality, 
	i.e.~$\mathcal{E}_t^n(\omega,x)\geq\mathcal{D}_t^n(\omega)+x$,
	fix some measure $Q\in\mathcal{M}_t^T(\omega)$ which we write as
	\[Q=Q_t\otimes Q'\]
	for a measure $Q_t$ on $\Omega_1$ and a Borel kernel
	$Q'\colon\Omega_1\to\mathfrak{P}(\Omega_{T-t-1})$.
	Then
	\[Q_t\in\mathcal{M}_t(0)\qquad\text{and}\qquad 
	Q'(\cdot)\in\mathcal{M}_{t+1}^T(\omega,\cdot)\quad\text{$Q_t$-almost surely}\]
	where $\mathcal{M}_t(0)=\mathcal{M}_t(Z)$ for the function $Z\equiv 0$.
	Indeed, first notice that
	\[ H(Q_t,\mathcal{P}_t(\omega)) + E_{Q_t}[H(Q'(\cdot),\mathcal{P}_{t+1}^T(\omega,\cdot))]
	=H(Q,\mathcal{P}_t^T(\omega))
	<+\infty \]
	by Lemma \ref{lem:entropie.sum.robust}, so that 
	\[H(Q'(\cdot),\mathcal{P}_{t+1}^T(\omega,\cdot))<+\infty\quad Q_t\text{-almost surely}.\]
	Similarly, we conclude that 
	$E_{Q'(\cdot)}[X(\omega,\cdot)^-+Y(\omega,\cdot) + |\Delta S_s(\omega,\cdot)|]<+\infty$ $Q_t$-almost surely
	for all $t+2\leq s\leq T$.
	Thus it holds
	\[Q'(\cdot)\in\mathcal{M}_{t+1}^T(\omega,\cdot)\quad\text{$Q_t$-almost surely}.\]
	But then it follows from the definition of $\mathcal{D}_{t+1}^n$ that
	\[E_{Q'(\cdot)}[X(\omega,\cdot)\wedge n] 
	- H(Q'(\cdot),\mathcal{P}_{t+1}^T(\omega,\cdot))
	\leq \mathcal{D}^n_{t+1}(\omega\otimes_t\cdot) \]
	$Q_t$-almost surely, so that 
	\begin{align*} 
	E_{Q_t}[\mathcal{D}^n_{t+1}(\omega\otimes_t\cdot)^-] 
	&\leq E_{Q_t}\big[\big(E_{Q'(\cdot)}[X(\omega,\cdot)\wedge n] 
	- H(Q'(\cdot),\mathcal{P}_{t+1}^T(\omega,\cdot))\big)^-\big]\\
	&\leq E_Q[X(\omega,\cdot)^-] + H(Q,\mathcal{P}_t^T(\omega))	
	<+\infty
	\end{align*}
	by Jensen's inequality and Lemma \ref{lem:entropie.sum.robust}.
	Therefore, one has $Q_t\in\mathcal{M}_t(0)$ and it follows that
	\begin{align*}
	&E_Q[X(\omega,\cdot)\wedge n]-H(Q,\mathcal{P}_t^T(\omega))\\
	&=E_{Q_t}[ E_{Q'(\cdot)}[ X(\omega,\cdot)\wedge n] 
	- H(Q'(\cdot),\mathcal{P}_{t+1}^T(\omega,\cdot))] - H(Q_t,\mathcal{P}_t(\omega))\\
	&\leq \sup_{R\in\mathcal{M}_t(0)}\big( 
	E_{R}[\mathcal{D}^n_{t+1}(\omega\otimes_t\cdot)]- H(R,\mathcal{P}_t(\omega))\big)
	=\mathcal{E}_t^n(\omega,x)-x
	\end{align*}
	and as $Q\in\mathcal{M}_t^T(\omega)$ was arbitrary, that indeed
	$\mathcal{D}_t^n(\omega)+x\leq \mathcal{E}_t(\omega,x)$.
	Coupled with the other inequality which was shown before, it holds
	$\mathcal{E}_t(\omega,x)=\mathcal{D}_t^n(\omega)+x$ and the proof is complete. 
\end{proof}

The following lemma will be important in the proof of the dynamic programming principle.
Since it was already shown that $\mathcal{E}_t(\omega,x)=\mathcal{D}_t(\omega)+x$,
the proof is almost one to one to the one for \cite[Lemma 3.7]{nutz2014utility}.
For the sake of completeness, a proof is given in the Appendix.

\begin{lemma}[\text{\cite[Lemma 3.7]{nutz2014utility}}]
\label{lem:existence.optimizer.measurable}
	For every $0\leq t\leq T-1$ and $x\in\mathbb{R}$, there exists a process $\vartheta^\ast\in\Theta$ such that
	\begin{align*}
	\mathcal{E}_s(\omega,x+(\vartheta^\ast\cdot S)_t^s(\omega))
	=\sup_{P\in\mathcal{P}_s(\omega)} \log E_P[\exp(\mathcal{E}_{s+1}(\omega\otimes_s\cdot,
	x+(\vartheta^\ast\cdot S)_t^{s+1}(\omega,\cdot) ))]
	\end{align*}	
	for all $t\leq s\leq T-1$ and $\omega\in\Omega_s$.
\end{lemma}

\begin{proof}[\text{\bf Proof of Theorem \ref{thm:multiperiod} -- Dynamic programming}]
	We turn to the proof of the dynamic programming principle, i.e.~we show that
	\begin{align}
	\label{eq:dyn.prog}
	C:=\inf_{\vartheta\in\Theta}\sup_{P\in\mathcal{P}_t^T(\omega)}
	\log E_P[\exp(X(\omega,\cdot)+x+ (\vartheta\cdot S)_t^T(\omega,\cdot))]
	=\mathcal{E}_t(\omega,x)
	\end{align}
	for all $x$, $\omega\in\Omega_t$, and $0\leq t\leq T-1$ and that the infimum over
	$\vartheta\in\Theta$ is attained.
	Again, fix some $x$, $t$, and $\omega\in\Omega_t$.
	By the first part of the proof of Theorem \ref{thm:multiperiod},
	i.e.~the part which focuses on duality,
	it holds $\mathcal{E}_t(\omega,x)=\mathcal{D}_t(\omega)+x$.
	Therefore $x$ can be subtracted on both sides of \eqref{eq:dyn.prog}
	and there is no loss of generality in assuming that $x=0$. This will lighten notation.

	First, we focus on the inequality $C\geq \mathcal{E}_t(\omega,x)$. 
	Fix some $\vartheta\in\Theta$, $P\in\mathcal{P}_t^T(\omega)$, and $Q\in\mathcal{M}_t^T(\omega)$.
	If 
	\[ C':=\log E_P[\exp(X(\omega,\cdot)+ (\vartheta\cdot S)_t^T(\omega,\cdot))]
	\geq E_Q[X(\omega,\cdot)] - H(Q,P),\]
	then the claim follows by taking the supremum over all those $Q$ and $P$,
	and in a second step the infimum over all $\vartheta\in\Theta$.
	To show this, one may assume that $C'$ and $H(Q,P)$ are finite, otherwise
	there is nothing to prove. Define 
	\[ Z:=X(\omega,\cdot)+(\vartheta\cdot S)_t^T(\omega,\cdot).\]
	Applying the elementary inequality 
	$ab\leq \exp(a)+b\log b$ to ``$a=Z^+$'' and ``$b=dQ/dP$'' yields
	\[ E_Q[Z^+]\leq E_P[\exp(Z^+)]+H(Q,P)\leq \exp(C')+1+H(Q,P)<+\infty.\]
	Therefore, it holds
	\[ E_Q[(\vartheta\cdot S)_t^T(\omega,\cdot)^+]
	\leq E_Q[X(\omega,\cdot)^-] + E_Q[Z^+]
	<+\infty \]
	by the definition of $\mathcal{M}_t^T(\omega)$.
	But then it follows from a result on local martingales (see \cite[Theorem 1 and 2]{jacod1998local}) that
	$(\vartheta\cdot S)_t^T(\omega,\cdot)$ is actually integrable with respect to $Q$ and has expectation 0.
	Hence $E_Q[Z^-]<+\infty$ and, therefore, Lemma \ref{lem:rep.exponential.dominated} yields
	\[C'=\log E_P[\exp(Z)] 
	\geq  E_Q[Z]  - H(Q,P)
	= E_Q[X(\omega,\cdot)] -H(Q,P)\]
	which is what we wanted to show.	

	We complete the proof by showing that $C\leq \mathcal{E}_t(\omega,0)$
	and that an optimal strategy $\vartheta^\ast\in\Theta$ exists.
	Let $\vartheta^\ast$ be the as in Lemma \ref{lem:existence.optimizer.measurable}, i.e.~such that
	\begin{align}
	\label{eq:optimal.H}
		\mathcal{E}_s(\omega,(\vartheta^\ast \cdot S)_t^s(\omega))
		=\sup_{P\in\mathcal{P}_s(\omega)} 
		\log E_P[\exp(\mathcal{E}_{s+1}(\omega\otimes_s\cdot,(\vartheta^\ast \cdot S)_t^{s+1}(\omega,\cdot))]
	\end{align}
	for all $t\leq s\leq T-1$. 
	Then $\vartheta^\ast$ is optimal and $C\leq \mathcal{E}_t(\omega,0)$.
	Indeed, let $P=P_t\otimes\cdots\otimes P_{T-1}\in\mathcal{P}_t^T(\omega)$ and
	fix some $t\leq s\leq T-1$. 
	Then it follows from \eqref{eq:optimal.H} that
	\begin{align*}
	&\log E_P\big[\exp( \mathcal{E}_s(\omega\otimes_t\cdot,(\vartheta^\ast \cdot S)_t^s(\omega,\cdot))\big]\\
	&=\log E_{P_t\otimes\cdots\otimes P_{s-1}}\big[\exp(
	\mathcal{E}_s(\omega\otimes_t\cdot,(\vartheta^\ast \cdot S)_t^s(\omega,\cdot))\big)\big]\\
	&\geq \log E_{P_t\otimes\cdots\otimes P_{s-1}}\Big[\exp\Big(\log 
	E_{P_s(\cdot)}\big[\exp\big(
	\mathcal{E}_{s+1}(\omega\otimes_t\cdot,(\vartheta^\ast \cdot S)_t^{s+1}(\omega,\cdot))\big)\big]\Big)\Big] \\
	&=\log E_P\big[\exp\big(\mathcal{E}_{s+1}(\omega\otimes_t\cdot,(\vartheta^\ast \cdot S)_t^{s+1}(\omega,\cdot))\big)\big],
	\end{align*}
	and an iteration yields
	\begin{align*}
	\mathcal{E}_t(\omega,0)
	&=\log E_P[\exp(\mathcal{E}_t(\omega,(\vartheta^\ast \cdot S)_t^t(\omega)))]\\
	&\geq \log E_P[\exp(\mathcal{E}_T(\omega\otimes_t\cdot,(\vartheta^\ast \cdot S)_t^T(\omega,\cdot)))]\\
	&=\log E_P[\exp( X(\omega,\cdot) + (\vartheta^\ast\cdot S)_t^T(\omega,\cdot))].
	\end{align*}
	As $P\in\mathcal{P}_t^T(\omega)$ was arbitrary, it holds
	$C\leq \mathcal{E}_t(\omega,x)$ and since $\vartheta^\ast\in\Theta$, it follows from the 
	previously shown inequality that $\vartheta^\ast$ is optimal.
\end{proof}

\subsection{The case with options}

Fix some function $Y\colon\Omega\to[0,+\infty)$ such that $-Y$ is upper semianalytic and recall that
$\mathcal{M}(Y):=\{Q\in\mathcal{M}: E_Q[Y]<+\infty\}$ and 
$\mathcal{M}_g(Y):=\{Q\in\mathcal{M}_g : E_Q[Y]<+\infty\}$, where $\mathcal{M}$ and $\mathcal{M}_g$
where defined in Section \ref{sec:main}.
Moreover, fix some Borel function $Z\colon\Omega\to\mathbb{R}$. 

We first claim that for every upper semianalytic function $X\colon\Omega\to\mathbb{R}$
bounded by $Z$, i.e.~$|X|\leq Z$, one has
\begin{align}
\label{eq:superhedge.entropy} 
\inf\{ m\in\mathbb{R} : m+(\vartheta\cdot S)_0^T\geq X\,\mathcal{P}\text{-q.s. for some }
\vartheta\in\Theta\}
&=\sup_{Q\in\mathcal{M}(Y)} E_Q[X]
\end{align}
in case of no options and, if $|g^i|\leq Z$ for $1\leq i\leq e$, then
\begin{align}
\label{eq:superhedge.0.in.ri} 
0\in\mathop{\mathrm{ri}} \{ E_Q[g^e] :  Q\in\mathcal{M}_{\hat{g}}(Y)\}
\end{align}
where $\hat{g}:=(g^1,\dots,g^{e-1})$ and also 
\begin{align}
\label{eq:superhedge.options.entropy} 
&\inf\Big\{m\in\mathbb{R} : 
\begin{array}{l}
m+(\vartheta\cdot S)_0^T+\alpha  g \geq X\,\mathcal{P}\text{-q.s.}\\
\text{for some }(\vartheta,\alpha)\in\Theta\times\mathbb{R}^e 
\end{array}\Big\}
=\sup_{Q\in\mathcal{M}_g(Y)} E_Q[X].
\end{align}
All these claims are proven in \cite{bouchard2015arbitrage} if one relaxes $\mathcal{M}$
in the sense that the relative entropy does not need to be finite.
In fact, Bouchard and Nutz deduce 
\eqref{eq:superhedge.options.entropy} from \eqref{eq:superhedge.0.in.ri}, and 
\eqref{eq:superhedge.0.in.ri} from \eqref{eq:superhedge.entropy};
see Theorem 4.9 as well as equation (5.1) and Theorem 5.1 in \cite{bouchard2015arbitrage}, respectively.
The same can be done here (with the exact same arguments as in \cite{bouchard2015arbitrage}), 
so we shall only give a (sketch of a) proof for \eqref{eq:superhedge.entropy}.
Consider first the one period case and define
\[\mathcal{C}':=\{ Q\in\mathfrak{P}(\Omega) : E_Q[|\Delta S| + Y]<+\infty
\text{ and } Q\ll P \text{ for some }P\in\mathcal{P} \},\]
and
$\mathcal{M}':=\{ Q\in\mathcal{C} : E_Q[\Delta S]=0\}$.
Then the following superhedging duality
\[ \inf\{m\in\mathbb{R} : m+h\Delta S\geq X\,\mathcal{P}\text{-q.s.} \text{ for some }h\in\mathbb{R}^d\}
=\sup_{Q\in\mathcal{M}'} E_Q[X],\]
see \cite[Theorem 3.4]{bouchard2015arbitrage}, is a consequence of the fact that
$0\in\mathop{\mathrm{ri}}\{E_Q[\Delta S] :Q\in\mathcal{C}'\}$;
see Lemma 3.5 and Lemma 3.6 in \cite{bouchard2015arbitrage}.
However, since
\[0\in\mathop{\mathrm{ri}} \{ E_Q[g] : Q\in\mathcal{C} \}\quad\text{for}\quad
\mathcal{C}=\{ Q \in\mathcal{C} : H(Q,\mathcal{P})<+\infty\}\]
by Lemma \ref{lem:fundamental.lem}, 
the same arguments as for \cite[Theorem 3.4]{bouchard2015arbitrage} show that
\[ \inf\{m\in\mathbb{R} : m+h\Delta S\geq X\,\mathcal{P}\text{-q.s.} \text{ for some }h\in\mathbb{R}^d\}
=\sup_{Q\in\mathcal{M}(Y)} E_Q[X],\]
in particular $\sup_{Q\in\mathcal{M}'} E_Q[X]=\sup_{Q\in\mathcal{M}(Y)} E_Q[X]$.
For the transition to the multiperiod case define recursively $m_T:=m_T':=X$ and
\[m_t'(\omega):=\sup_{Q\in\mathcal{M}_t'(\omega)} E_Q[m_{t+1}'(\omega,\cdot)]
\quad\text{and}\quad
m_t(\omega)=\sup_{Q\in\mathcal{M}_t^Z(\omega)} E_Q[m_{t+1}(\omega,\cdot)],\]
for $0\leq t\leq T-1$ and $\omega\in\Omega_t$, where
\begin{align*} 
\mathcal{M}_t'(\omega)&:=\{Q\in\mathfrak{P}(\Omega_1): Q\ll P \text{ for some }P\in\mathcal{P}_t(\omega)
\text{ and } E_Q[\Delta S_{t+1}(\omega,\cdot)]=0\},\\
\mathcal{M}_t^Z(\omega)&:=\{Q\in\mathcal{M}_t'(\omega) : E_Q[Z]+E_Q[m_{t+1}(\omega,\cdot)^-]+H(Q,\mathcal{P}_t(\omega))<+\infty\}
\end{align*}
and $Z\colon\Omega_1\to[0,+\infty)$ an arbitrary universally measurable function.
A backward induction shows that $m_t=m_t'$ for each $t$.
Moreover, following the exact same arguments as in the part of the proof for Theorem \ref{thm:multiperiod}
which focuses on duality, one can show that 
$m_t(\omega)=\sup_{Q\in\mathcal{M}_t^T(\omega)} E_Q[X(\omega,\cdot)]$
where we recall
\[ \mathcal{M}_t^T(\omega)
=\bigg\{ Q\in\mathfrak{P}(\Omega_{T-t}): 
\begin{array}{l} 
(S_s(\omega,\cdot))_{t\leq s \leq T} \text{ is a $Q$-martingale and}\\
E_Q[X(\omega,\cdot)^-+Y(\omega,\cdot)] + H(Q,\mathcal{P}_t^T(\omega))<+\infty
\end{array}\bigg\} \]
so that $\mathcal{M}_0^T=\mathcal{M}(Y)$.
Since it is shown in (or rather within the proof of) 
\cite[Lemma 4.13]{bouchard2015arbitrage} that
\[\inf\{ m\in\mathbb{R} : m+(\vartheta\cdot S)_0^T\geq X\,\mathcal{P}\text{-q.s. for some }
\vartheta\in\Theta\}
=m_0',\]
the claim follows from $m_0'=m_0=\sup_{Q\in\mathcal{M}(Y)}E_Q[X]$.

\begin{proof}[\text{\bf Proof of Theorem \ref{thm:main.options}}]
	The proof is an induction over $e$.
	For $e=0$, the statement is a special case of Theorem \ref{thm:multiperiod}, 
	so assume that both claims (duality and existence) are true for $e-1\geq0$.
	By assumption, there is a Borel function $Z$ such that $|X|+|g^i|\leq Z$ for every $1\leq i\leq e$.
	Using the induction hypothesis, it follows that
	\begin{align}
	\label{eq:options.infinf}
	&\inf_{(\vartheta,\alpha)\in\Theta\times\mathbb{R}^e}
	\sup_{P\in\mathcal{P}} \log E_P[\exp(X+(\vartheta\cdot S)_0^T+\alpha g)]\\
	&=\inf_{\beta \in\mathbb{R}}\min_{(\vartheta,\hat{\alpha})\in\Theta\times\mathbb{R}^{e-1}} 
	\sup_{P\in\mathcal{P}} \log E_P[\exp(X+(\vartheta\cdot S)_0^T+\hat{\alpha}\hat{g}+\beta g^e)]
	\label{eq:options.inf.min} \\
	&=\inf_{\beta \in\mathbb{R}} \sup_{Q\in\mathcal{M}_{\hat{g}}} \big( E_Q[X]+\beta E_Q[g^e] -H(Q,\mathcal{P}) \big)
	=\inf_{\beta \in\mathbb{R}} \sup_{Q\in\mathcal{M}_{\hat{g}}} J(Q,\beta )\nonumber 
	\end{align}
	where $\hat{g}=(g^1,\dots,g^{e-1})$ and
	\[J\colon \mathcal{M}_{\hat{g}}\times\mathbb{R}\to\mathbb{R},
	\quad (Q,\beta )\mapsto E_Q[X]+\beta E_Q[g^e] -H(Q,\mathcal{P}).\] 
	It is already shown that
	$0\in\mathop{\mathrm{ri}} \{ E_Q[g^e] : Q\in\mathcal{M}_{\hat{g}}\}$,
	see \eqref{eq:superhedge.0.in.ri},
	which can be used exactly as in the proof of Theorem \ref{thm:1peroiod} to prove that 
	\begin{align}
	\label{eq:options.minimax}
	 \inf_{|\beta |\leq n} \sup_{Q\in\mathcal{M}_{\hat{g}}} J(Q,\beta)
	=\inf_{\beta\in\mathbb{R}} \sup_{Q\in\mathcal{M}_{\hat{g}}} J(Q,\beta)
	=\sup_{Q\in\mathcal{M}_{\hat{g}}} \inf_{\beta\in\mathbb{R}} J(Q,\beta)
	\end{align}
	for some $n\in\mathbb{N}$;  see \eqref{eq:for.existence.minimizer} for the first,
	and the text below \eqref{eq:supinfcondition} for the second equality.
	Hence
	\begin{align*}
	&\inf_{(\vartheta,\alpha)\in\Theta\times\mathbb{R}^e} 
	\sup_{P\in\mathcal{P}} \log E_P[\exp(X+(\vartheta\cdot S)_0^T+\alpha g)]\\
	&=\inf_{\beta \in\mathbb{R}} \sup_{Q\in\mathcal{M}_{\hat{g}}} J(Q,\beta)
	=\sup_{Q\in\mathcal{M}_{\hat{g}}} \inf_{\beta \in\mathbb{R}} J(Q,\beta)
	=\sup_{Q\in\mathcal{M}_g} \big( E_Q[X]-H(Q,\mathcal{P})\big)
	\end{align*}
	showing that duality holds.
	The first equality in \eqref{eq:options.minimax} together with the lower semicontinuity of 
	$\sup_{Q\in\mathcal{M}_{\hat{g}}} J(Q,\cdot)$
	imply that there is some $\beta^\ast\in\mathbb{R}$ such that 
	\[\sup_{Q\in\mathcal{M}_{\hat{g}}} J(Q,\beta^\ast)
	=\inf_{\beta\in\mathbb{R}}\sup_{Q\in\mathcal{M}_{\hat{g}}} J(Q,\beta).\]
	For this $\beta^\ast$, the induction hypotheses \eqref{eq:options.inf.min} guarantees the existence of an optional
	strategy $(\vartheta^\ast,\hat{\alpha}^\ast)\in\Theta\times\mathbb{R}^{e-1}$ showing that
	$(\vartheta^\ast,\alpha^\ast)\in\Theta\times\mathbb{R}^e$ is optimal for 
	\eqref{eq:options.infinf}, where $\alpha^\ast:=(\hat{\alpha}^\ast,\beta^\ast)$.
	This completes the proof.
\end{proof}

\begin{proof}[\text{\bf Proof of Theorem \ref{thm:limit.superhedg}}]
	Since $\Theta$ and $\mathbb{R}^e$ are vector-spaces, it follows from 
	Theorem \ref{thm:main.options} that
	\begin{align*}
	\pi_\gamma(X)
	&=\inf_{(\vartheta,\alpha)\in\Theta\times\mathbb{R}^e}\sup_{P\in\mathcal{P}}
	\frac{1}{\gamma}\log E_P\big[\exp\big(\gamma X + (\vartheta\cdot S)_0^T + \alpha g\big)\big]\\
	&= \frac{1}{\gamma}\sup_{Q\in\mathcal{M}_g} \big( E_Q[\gamma X] -H(Q,\mathcal{P}) \big)
	=\sup_{Q\in\mathcal{M}_g} \big( E_Q[X] -\frac{1}{\gamma}H(Q,\mathcal{P}) \big).
	\end{align*} 
	This formula implies both that $\pi_\gamma$ is increasing in $\gamma$ and, 
	by interchanging the suprema over $\gamma$ and $Q$,  that
	$\sup_\gamma \pi_\gamma(X)=\sup_{Q\in\mathcal{M}_g} E_Q[X]$. 
	The latter term coincides by \eqref{eq:superhedge.options.entropy} 
	with $\pi(X)$, hence the proof is complete.
\end{proof}

\appendix
\section{Technical proofs}
\label{sec:app.proofs}

We start by proving Remark \ref{rem:main.assumptions}, Remark \ref{rem:main.discussion},
and the statements of Section \ref{sec:examples}.

\begin{proof}[\text{\bf Proof of Remark \ref{rem:main.assumptions}}]
	1)
	Let $T=d=1$, $\Omega=\mathbb{R}$, $S_0=0$, $S_1(\omega)=\omega$,
	and define $\mathcal{P}=\mathop{\mathrm{conv}}\{\delta_x: x\in[0,1]\}$
	so that NA$(\mathcal{P})$ fails.
	Then in  \eqref{eq:optim.problem} the left-hand side is always larger or equal than $X(0)$,
	and the right-hand side equals $X(0)$ since $\mathcal{M}=\{\delta_0\}$.	
	For the choice $X=-1_{\{0\}}$, a short computation yields that the left-hand side actually equals 0, 
	showing that there is a gap.   

	2)
	Let again $T=d=1$, $\Omega=\mathbb{R}$, $S_0=0$, $S_1(\omega)=\omega$,
	and define $\mathcal{P}=\mathop{\mathrm{conv}}\{\delta_{-1},\delta_x: x\in(0,1]\}$.
	Then NA$(\mathcal{P})$ holds true and every martingale measure 
	$Q$ with $H(Q,\mathcal{P})<+\infty$ satisfies
	$Q(\{-1\})>0$. In particular, for $X:=-\infty 1_{\{-1\}}$ 
	the right-hand side of \eqref{eq:optim.problem} equals $-\infty$
	while the right-hand side satisfies
	\[ \inf_{h\in\mathbb{R}} \sup_{P\in\mathcal{P}} \log E_P[\exp(X+h\Delta S)]
	\geq \inf_{h\in\mathbb{R}} \lim_{x\downarrow 0} \log \frac{\exp(-\infty)+ \exp(hx)}{2}
	=\log\frac{1}{2} \]
	as $(\delta_{-1}+\delta_x)/2\in\mathcal{P}$ for every $x\in(0,1]$.
	To see that an optimal strategy $h\in\mathbb{R}$ needs not to exists, 
	take the same $X$ but let $\mathcal{P}=\{(\delta_{-1}+\delta_{1})/2\}$.
\end{proof}

\begin{proof}[\text{\bf Proof of Remark \ref{rem:main.discussion}}]
	We claim that for any probability $P$ satisfying the classical no-arbitrage,
	it is possible to construct $\mathcal{P}_t$ such that $\mathcal{P}=\{P\}$
	and {\rm NA}$(\mathcal{P}_t(\omega))$ holds for every $t$ and $\omega\in\Omega_t$
	and only sketch the proof. Write $P=P_0\otimes\cdots\otimes P_{T-1}$ for the kernels
	$P_t$ from Remark \ref{rem:kernels.measurable.and.derivative.measurable}
	and define $N_t:=\{\omega\in\Omega_t : \mathrm{NA}(P_t(\omega)) \text{ fails}\}$.
	Then it holds 
	\[ N_t=\pi \Big\{ (\omega,h)\in\Omega_t\times\mathbb{R}^d :
	\begin{array}{l} 
	P_t(\omega)(h\Delta S_{t+1}(\omega,\cdot)\geq 0)=1 \text{ and}\\
	P_t(\omega)(h\Delta S_{t+1}(\omega,\cdot)>0)>0 
	\end{array} \Big\} \]
	and by the classical fundamental theorem of asset pricing
	\[ N_t^c=\pi\{ (\omega,Q)\in\Omega_t\times\mathfrak{P}(\Omega_1) : 
	E_Q[\Delta S_{t+1}(\omega,\cdot)]=0 \text{ and } Q\sim P_t(\omega) \}.\]
	In both cases $\pi$ denotes the projection onto $\Omega_t$.
	It can be shown that both sets, which the projection acts on, are Borel. 
	Thus $N_t$ and $N_t^c$ are analytic sets.
	Now define $\mathcal{P}_t(\omega)=\{P_t(\omega)\}$ if $\omega\in N_t^c$ and
	$\mathcal{P}_t(\omega)=\{\delta_{S_t(\omega)}\}$ else.
	Then $\mathcal{P}_t$ has analytic graph and since $N_t$ is a zero set under $P$, 
	it follows that $\mathcal{P}=\{P\}$. 
\end{proof}

\begin{proof}[\text{\bf Proofs for Section \ref{sec:examples}}]

(a)
The graphs of \eqref{eq:P.robust.general}, \eqref{eq:Rt.dist}, and \eqref{eq:Rt.phitn} 
are Borel: For $\mathcal{P}_t$ defined by \eqref{eq:P.robust.general}, notice that  
\[ g\colon\Omega_t\times\mathfrak{P}(\Omega_1)\to[0,+\infty],
\quad(\omega,R)\mapsto \mathop{\mathrm{dist}}(R,P_t(\omega))/\varepsilon_t(\omega)\]
is Borel, hence $\mathop{\mathrm{graph}}\mathcal{P}_t=\{g\leq 1\}$ is Borel, and therefore analytic.
The proofs for \eqref{eq:Rt.dist} and \eqref{eq:Rt.phitn} are analogue.

(b)
If $\mathcal{R}_t$ has analytic graph, then so do $\hat{\mathcal{P}}_t$ and $\mathcal{P}_t$:
Define 
\[ g\colon\Omega_t\times\mathfrak{P}(\Omega_1)\to\Omega_t\times\mathfrak{P}(\Omega_1),
\quad(\omega,P)\mapsto (\omega,P\circ f_t(\omega,\cdot)^{-1})\]
and notice that $g$ is Borel by Lemma \ref{lem:integral.is.measurable}
and \cite[Proposition 7.26]{bertsekas1978stochastic}. 
Therefore, $\mathop{\mathrm{graph}}\hat{\mathcal{P}}_t=g(\mathop{\mathrm{graph}}\mathcal{R}_t)$ is an analytic
set, as the image of such set under a Borel function.
As for $\mathcal{P}_t$, define the Borel function
\begin{align*}
g_n\colon \big( (\Omega_t\times\mathfrak{P}(\Omega_1))^n\cap \Delta_n\big)\times  C_n&\to \Omega_t\times\mathfrak{P}(\Omega_1),\\
((\omega^i,P^i)_i,\lambda)&\mapsto (\omega^1,\lambda_1 P^1+\cdots+\lambda_nP^n)
\end{align*}
for every $n\in\mathbb{N}$, where
\begin{align*}
\Delta_n&:=\{ (\omega^i,P^i)_i\in (\Omega_t\times\mathfrak{P}(\Omega_1))^n : \omega^1=\omega^i \text{ for }1\leq i\leq n\},\\
C_n&:=\{ \lambda\in[0,+\infty) : \lambda_1+\dots+\lambda_n=1 \}.
\end{align*}
Therefore, as the countable union of the images under Borel functions of analytic sets,
\[ \mathop{\mathrm{graph}}\mathcal{P}_t
=\bigcup_n g_n\Big( \big(\mathop{\mathrm{graph}}\hat{\mathcal{P}}_t)^n\cap \Delta_n)\times C_n\big)\Big) \]
is again an analytic set.

(c)
On the no-arbitrage condition.
We only prove the claim for the Wasserstein distance of order $p$, i.e.~$\mathcal{R}_t$ given by
\eqref{eq:Rt.dist}, the proof for $\mathcal{R}_t$ given by \eqref{eq:Rt.phitn} works similar.
Fix $\omega\in\Omega_t$ and let $h\in\mathbb{R}$ such that 
$h\Delta S_{t+1}(\omega,\cdot)\geq 0$ $\mathcal{P}_t(\omega)$-q.s. 
If $f_t(\omega,\mathbb{R})=\{S_t(\omega)\}$, 
then trivially $h\Delta S_{t+1}(\omega,\cdot)= 0$ $\mathcal{P}_t(\omega)$-q.s. 
Otherwise, there are $y^\pm\in\mathbb{R}$ such that $\pm f(\omega,y^\pm)>0$ by assumption.
Now define $R^\pm:= \lambda^\pm \delta_{y^\pm}+(1-\lambda^\pm) \mathop{\mathrm{law}}X$, where 
$\lambda^\pm:=1\wedge 1/(\mathop{\mathrm{dist}}(\delta_{y^\pm},\mathop{\mathrm{law}}X)\varepsilon_t(\omega))$
is strictly positive since $X$ has finite $p$-th moment. By convexity,
\[\mathop{\mathrm{dist}}(R^\pm,\mathop{\mathrm{law}}X)
\leq \lambda^\pm \mathop{\mathrm{dist}}(\delta_{\pm x},\mathop{\mathrm{law}}X)
+(1-\lambda^\pm)\mathop{\mathrm{dist}}(\mathop{\mathrm{law}}X,\mathop{\mathrm{law}}X)
\leq\varepsilon_t(\omega)\]
so that $R^\pm\in\mathcal{R}_t(\omega)$.
Hence $hf(\omega,y^\pm)\geq 0$, which in turn implies $h=0$.

(d)
The Binomial and Black-Scholes model.
A computation as in a) shows that the graph of $\Phi_t$ defined by
\[ \Phi_t(\omega):=\big\{ (q,a,b) : p\in[\underline{p}_t(\omega),\overline{p}_t(\omega)],
	a\in[\underline{a}_t(\omega), \overline{a}_t(\omega)], b\in[\underline{b}_t(\omega),\overline{b}_t(\omega)] \big\} \]
is an analytic set.
Since 
\[g\colon\Omega_t\times\mathbb{R}^3\to\Omega_t\times\mathfrak{P}(\Omega_1),
\quad (\omega,p,a,b)\mapsto (\omega, p\delta_a+(1-p)\delta_b)\]
is continuous, it follows that $\mathop{\mathrm{graph}}\mathcal{R}_t=g(\mathop{\mathrm{graph}}\Phi_t)$
is an analytic set.
The proof for the Black-Scholes model works similar.
\end{proof}

The following lemma is related to \cite[Lemma 3.29]{follmer2011stochastic},
where $X$ is assumed to be bounded. 

\begin{lemma}
	\label{lem:rep.exponential.dominated}
	Let $X\colon\Omega\to\mathbb{R}$ be measurable and let $P\in\mathfrak{P}(\Omega)$.
	Then one has
	\[\log E_P[\exp(X)]=\sup_{Q\in\mathcal{A}} ( E_Q[X]-H(Q,P) )\]
	where $\mathcal{A}:=\{Q\in\mathfrak{P}(\Omega) : H(Q,P)+ E_Q[X^-]<+\infty\}$.
\end{lemma}
\begin{proof}
	For each natural number $n$, define $Q_n$ by 
	\[\frac{dQ_n}{dP}:= \frac{\exp(X\wedge n)}{E_P[\exp(X\wedge n)]}.\]
	Then $Q_n$ is equivalent to $P$ and since
	$\exp(X\wedge n) X^-\leq 1$, it follows that $X^-$ is integrable 
	with respect to $Q_n$.
	By equivalence of $P$ and $Q_n$, one can write
	\[ \frac{dQ}{dP}=\frac{dQ}{dQ_n}\frac{dQ_n}{dP}
	\qquad\text{for any }Q\in\mathcal{A}.\] 
	Applying Jensen's inequality to the convex function
	$[0,\infty)\to[-1,\infty),\,x\mapsto x\log x$ with
	``$x=dQ/dQ_n$" yields
	\begin{align*} 
	H(Q,P) 
	&=E_{Q_n}\Big[\frac{dQ}{dQ_n}\log \frac{dQ}{dQ_n}\Big]+E_Q\Big[\log\frac{dQ_n}{dP}\Big]\\
	&\geq E_Q\Big[\log\frac{dQ_n}{dP}\Big]
	=E_Q[X\wedge n] - \log E_P[\exp(X\wedge n)]
	\end{align*}
	with equality if (and only if) $Q=Q_n$.
	Since the right-hand side is finite for $Q=Q_n$, 
	it follows that $H(Q_n,P)<+\infty$ and, therefore, $Q_n\in\mathcal{A}$.
	Rearranging the terms which appear in the inequality above yields
	\[\log E_P[\exp(X\wedge n)]\geq E_Q[X\wedge n]-H(Q,P) \]
	for all $Q\in\mathcal{A}$ with equality for $Q=Q_n\in\mathcal{A}$.
	This shows the claim if $X$ were bounded.
	The general case follows by letting $n$ tend to infinity.
	Indeed, since the set $\mathcal{A}$ does not depend on $n$,
	we can interchange two suprema and conclude that
	\[\log E_P[\exp(X)]
	=\sup_n \sup_{Q\in\mathcal{A}} (E_Q[X\wedge n]-H(Q,P))
	=\sup_{Q\in\mathcal{A}}(E_Q[X]-H(Q,P)).\]
	The use of the monotone  convergence theorem in the last step
	was justified because $E_Q[X^-]<+\infty$ for every $Q\in\mathcal{A}$.
\end{proof}

\begin{lemma}
\label{lem:absolute.continuity.kernels}
	Let $V$ and $W$ be two Polish spaces and $P,Q\in\mathfrak{P}(V\times W)$ 
	with representation $P=\mu\otimes K$, $Q=\mu'\otimes K'$ for measures $\mu,\mu'\in\mathfrak{P}(V)$
	and universally measurable kernels $K,K'\colon V\to\mathfrak{P}(W)$.
	Then one has
	\[Q\ll P \quad\text{if and only if}\quad  \mu'\ll\mu \text{ and } K'(v)\ll K(v) \]
	for $\mu'$-almost every $v$.
\end{lemma}
\begin{proof}
	If $\mu'\ll\mu$ and $K'(v)\ll K(v)$ for $\mu'$-almost every $v$,
	it follows from the definition that $Q\ll P$.
	Indeed, for any Borel set $A\subset V\times W$ such that
	$0=P(A)=E_{\mu(dv)}[K(v)(A_v)]$, it holds $Q(A)=E_{\mu'(dv)}[K'(v)(A_v)]=0$.
	Here, $A_v:=\{w\in W: (v,w)\in A\}$.

	The other direction needs more work. The idea is to show that
	the generalized Radon-Nikodym derivative (see e.g.~\cite[Theorem A.13]{follmer2011stochastic})
	is measurable with respect to the kernels. 
	Assume that $Q\ll P$ and first notice that $\mu'\ll\mu$. 
	If this were not the case, then $\mu'(A)>0$ while $\mu(A)=0$ for some Borel set $A\subset V$
	which implies $Q(A\times W)=\mu'(A)>0$ but $P(A\times W)=0$.
	We proceed to show the absolute continuity of the kernels.
	Notice that the mapping 
	\[ \mathfrak{P}(W)\times\mathfrak{P}(W)\times W,\quad 
	(R',R,w)\mapsto \frac{dR'}{dR}(w)\]
	can be shown to be Borel, where $dR'/dR$ denotes the Radon-Nikodym derivative of the
	absolutely continuous part of $R'$ with respect to $R$.
	This result, due to Doob, can be found e.g.~in 
	\cite[Theorem V.58]{dellacherie2011probabilities} and the subsequent remark.
 	Hence 
	\[V\times W \to \mathfrak{P}(W)\times\mathfrak{P}(W)\times W,\quad
	(v,w)\mapsto (K(v),K'(v),w)\] 
	is universally measurable.
 	Thus, since $(R,R')\mapsto (R+R')/2$ is Borel, it follows that
	$Z\colon V\times W\to[0,+\infty]$,
	\[	Z(v,w):=	\frac{d K'(v)}{d(K(v)+K'(v))/2}(w)\Big(\frac{d K(v)}{d(K(v)+K'(v))/2}(w)\Big)^{-1}\]
	is universally measurable, with the convention $x/0:=+\infty$ for all $x\geq 0$. 
	A straightforward computation as in \cite[Theorem A.13]{follmer2011stochastic} yields
	$K(v)(Z(v,\cdot)=+\infty)=0$,
	\[ K'(v)(B)=K'(v)(B\cap\{Z(v,\cdot)=+\infty\}) + E_{K(v)}[1_BZ(v,\cdot)] \]
	for any universally measurable set $B\subset W$, and as a consequence that
	\[ K'(v)\ll K(v)\quad\text{if and only if}\quad K'(v)(Z(v,\cdot)=+\infty)=0.\]
	Heading for a contradiction, assume that the set of all such $v$ 
	has not full $\mu'$ measure and define the universally measurable set
	\[A:=\{(v,w): Z(v,w)=+\infty\}.\] 
	Then 
	\[Q(A)
	=E_{\mu'(dv)}[K'(v)(Z(v,\cdot)=+\infty)]
	>0,\]
	while on the other hand
	$P(A)=E_{\mu(dv)}[ K(v)(Z(v,\cdot)=+\infty)]=0$.
	This contradicts the absolute continuity of $Q$ with respect to $P$.
\end{proof}

\begin{proof}[\text{\bf Proof of Lemma \ref{lem:entropie.sum}}]
	The goal is to show that
	\begin{align}
	\label{eq:entropie.sum} 
	H(Q,P)=\sum_{s=t}^{T-1} E_Q[H(Q_s(\cdot),P_s(\cdot))].
	\end{align}

	(a)
	We first comment on the measurability of terms appearing later.
	Fix some $t\leq s\leq T-1$ and notice as in the proof of 
	Lemma \ref{lem:absolute.continuity.kernels} that
	\[\Omega_{s-t}\to\mathfrak{P}(\Omega_1)\times\mathfrak{P}(\Omega_1),\quad
	\bar{\omega}\mapsto (Q_s(\bar{\omega}),P_s(\bar{\omega}))\] 
	is universally measurable.
	Since the entropy $H$ is Borel by Lemma \ref{lem:H.is.borel}, one can check that 
	\[\Omega_{s-t}\to[0,+\infty],\quad \bar{\omega}\mapsto H(Q_s(\bar{\omega}),P_s(\bar{\omega}))\]
	is universally measurable.
	Similar, using Doob's result on the measurability of the Radon-Nikodym derivative as
	in the proof of Lemma \ref{lem:absolute.continuity.kernels}, it follows that
	\[ \Omega_{s-t}\times\Omega_1\to[0,+\infty],\quad
	(\bar{\omega},\omega')\mapsto\frac{dQ_s(\bar{\omega})}{dP_s(\bar{\omega})}(\omega') \]
	is universally measurable.
	Moreover, by \cite[Lemma 7.29]{bertsekas1978stochastic},
	the same holds true if $\bar{\omega}\in\Omega_{s-t}$ is fixed in the above
	mapping, and the latter is considered as a function of $\omega'$.
	
	(b)
	A direct application of Lemma \ref{lem:absolute.continuity.kernels} shows that
	\[ Q\ll P \quad\text{if and only if}\quad
	Q_s\ll P_s \quad Q_t\otimes\cdots\otimes Q_{s-1}\text{-almost surely}\]
	for all $t\leq s\leq T-1$, where in case $s=t$ the above should be understood as
	$Q_t\ll P_t$.
	This implies that whenever $Q$ is not absolutely continuous with respect to $P$, then 
	both sides in \eqref{eq:entropie.sum} are equal to $+\infty$. 
	Hence we may assume that $Q\ll P$. 
	Then $dQ/dP$ can be expressed as the product of $d Q_s(\cdot)/d P_s(\cdot)$, where
	$s$ ranges from $t$ to $T-1$.
	Therefore, for any $t\leq s\leq T-1$, it follows that
	\[E_Q\Big[ \Big(\log \frac{d Q_s(\cdot)}{d P_s(\cdot)}\Big)^-\Big]
	=E_{Q_t\otimes\cdots\otimes Q_{s-1}}\Big[
	E_{P_s(\cdot)}\Big[ \frac{d Q_s(\cdot)}{d P_s(\cdot)}\Big(\log \frac{d Q_s(\cdot)}{d P_s(\cdot)}\Big)^- \Big]\Big]
	\leq 1,\]	
	where the last inequality holds since $x(\log x)^-\leq 1$ for all $x\geq 0$.
	By integrability, the same steps may be repeated without the negative parts
	so that
	\[ H(Q,P)
	=E_Q\Big[\log\frac{dQ}{dP}\Big]
	=\sum_{s=t}^{T-1}E_Q\Big[\log \frac{dQ_s(\cdot)}{dP_s(\cdot)} \Big]
	=\sum_{s=t}^{T-1} E_Q[H(Q_s(\cdot),P_s(\cdot))] \]
	as claimed.
\end{proof}

\begin{proof}[\text{\bf Proof of Lemma \ref{lem:existence.optimizer.measurable}}]
	First, we claim that for any $\vartheta\in\Theta$ and $0\leq t\leq T-1$,
	there exists a universally measurable mapping
	$\hat{h}_t\colon \Omega_t\to\mathbb{R}^d$ such that
	\begin{align}
	\label{eq:optimal.h.hat} 
	&\mathcal{E}_t(\omega,x+(\vartheta\cdot S)_0^t(\omega))\\
	&=\sup_{P\in\mathcal{P}_t(\omega)}\log E_P[\exp( \mathcal{E}_{t+1}(\omega\otimes_t\cdot,
	x+ (\vartheta\cdot S)_0^t(\omega)+\hat{h}_t(\omega)\Delta S_{t+1}(\omega,\cdot)))]\nonumber
	\end{align}
	for all $\omega\in\Omega_t$.
	To that end, fix some $\vartheta\in\Theta$, $0\leq t\leq T-1$,
	and recall that $\mathcal{F}_t$ was defined as the universal completion
	of the Borel $\sigma$-field on $\Omega_t$.
	From the first part of the proof of Theorem \ref{thm:multiperiod},
	we already know that $\mathcal{E}_t(\omega,x)=\mathcal{D}_t(\omega)+x$
	for all $\omega\in\Omega_t$ and $x\in\mathbb{R}$ and that $\mathcal{D}_t$
	is upper semianalytic, in particular $\mathcal{F}_t$-measurable.
	This implies that $\mathcal{E}_t$ is 
	$\mathcal{F}_t\otimes\mathcal{B}(\mathbb{R})$-measurable.
	Define the function
	\[\phi(\omega,x,h):=\sup_{P\in\mathcal{P}_t(\omega)}
	\log E_P[\exp( \mathcal{D}_{t+1}(\omega\otimes_t\cdot) +x + h\Delta S_{t+1}(\omega,\cdot))].\]
	For fixed $x$ and $h$, it follows from
	\cite[Proposition 7.47]{bertsekas1978stochastic}
	(as in the first first part of the proof of Theorem \ref{thm:multiperiod})
	that $\phi(\cdot,x,h)$ is upper semianalytic.
	Moreover, for fixed $\omega$, an application of Fatou's lemma 
	(as in part (b) of the proof of Theorem \ref{thm:1peroiod})
	shows that 
	$\phi(\omega,\cdot,\cdot)$ is lower semicontinuous.
	Therefore, we can conclude by \cite[Lemma 4.12]{bouchard2015arbitrage}
	that $\phi$ is 
	$\mathcal{F}_t\otimes\mathcal{B}(\mathbb{R})\otimes\mathcal{B}(\mathbb{R}^d)$-measurable.
	Now fix $x\in\mathbb{R}$ and define the set-valued mapping
	\[\Phi(\omega):=\{h\in\mathbb{R}^d : \phi(\omega,x+ (\vartheta\cdot S)_0^t(\omega),h)
	=\mathcal{E}_t(\omega,x+(\vartheta\cdot S)_0^t(\omega))\}.\]
	By Theorem \ref{thm:1peroiod} it holds  $\Phi(\omega)\neq\emptyset$ and by the above its
	graph is in $\mathcal{F}_t\otimes\mathcal{B}(\mathbb{R}^d)$.
	Hence it follows by Theorem 5.5 in \cite{leese1978measurable}
	or rather the corollary and scholim after, that $\Phi$ admits an $\mathcal{F}_t$-measurable
	selector $\hat{h}_t$.
	
	To conclude the proof of the lemma, define 
	$\vartheta^\ast_s:=0$ for $s\leq t$, let $\hat{h}_t$ be an optimal strategy for
	time $t$ and define $\vartheta^\ast_{t+1}:=\hat{h}_t$.
	By the above, there is a universally
	measurable mapping $\hat{h}_{t+1}\colon\Omega_{t+1}\to\mathbb{R}$ 
	such that \eqref{eq:optimal.h.hat} holds for $t+1$.
	Define $\vartheta^\ast_{t+2}:=\hat{h}_{t+1}$.
	Proceeding in a recursive matter until $t=T$, we construct $\vartheta^\ast\in\Theta$ which 
	fulfills the requirements of the lemma.
\end{proof}

\section{Analytic sets}
\label{sec:app.analytic}
We briefly recall the used terminology and give a short overview on the theory of analytic sets;
for more details see e.g.~Chapter 7 in the book of 
Bertsekas and Shreve \cite{bertsekas1978stochastic}.
Throughout, fix two Polish spaces $V$ and $W$.
A subset of a Polish space is called analytic, if it is the image of a Borel 
set of another Polish space under a Borel function. 
Similarly, a function $f\colon V\to[-\infty,+\infty]$
is upper semianalytic, if $\{f\geq c\}\subset V$ is an analytic set for every real number $c$.
Further define $\mathcal{B}(V)$ to be the Borel $\sigma$-field on $V$ and
$\mathfrak{P}(V)$ to be the set of all probability measures on $\mathcal{B}(V)$.
The set $\mathfrak{P}(V)$ is endowed with the weak topology induced by all continuous bounded functions,
i.e.~$\sigma(\mathfrak{P}(V),C_b(V))$. Then $\mathfrak{P}(V)$ becomes a Polish space itself.
The set of universally measurable subsets of $V$ is defined as 
$\bigcap \{ \mathcal{B}(V)^P : P\in\mathfrak{P}(V)\}$, where 
$\mathcal{B}(V)^P$ is the completion of $\mathcal{B}(V)$ with respect to the probability $P$.
A function $f\colon V\to W$ is said to be universally measurable, 
if $\{f\in B\}$ is universally measurable for
every $B\in\mathcal{B}(W)$.
It follows from the definition that every Borel set is analytic, and from
Lusin's theorem (see \cite[Proposition 7.42]{bertsekas1978stochastic})
that every analytic set is universally measurable.
The same of course holds true if we replace sets by functions in the previous sentence.
A set-valued function $\Psi\colon V\to W$ is said to have analytic graph, if
\[ \mathop{\mathrm{graph}} \Psi
:=\{ (v,w) : v\in V,  w\in\Psi(v) \}\subset V\times W \]
is an analytic set. Finally, given a set $\mathcal{P}\subset\mathfrak{P}(V)$, 
a set $N\subset V$ is said to be $\mathcal{P}$-polar if $P(N)=0$ for all $P\in\mathcal{P}$.
Similarly, a property is said to hold $\mathcal{P}$-quasi surely (q.s.~for short),
if it holds outside a $\mathcal{P}$-polar set.

One can readily verify that $(v,P(dw))\mapsto E_P[X(v,\cdot)]$ is continuous, 
whenever $X\colon V\times W\to\mathbb{R}$ 
is uniformly continuous and bounded. The following lemma generalizes this.

\begin{lemma}[\text{\cite[Proposition 7.29/7.46/7.48]{bertsekas1978stochastic}}]
	\label{lem:integral.is.measurable}
	Let $X\colon V\times\mathfrak{P}(W)\times W\to[-\infty,+\infty]$
	be Borel / upper semianalytic / universally measurable. Then the mapping
	$V\times\mathfrak{P}(W)\to[-\infty,+\infty]$,
	$(v,P)\mapsto E_P[X(v,P,\cdot)]$
	is Borel / upper semianalytic / universally measurable.
\end{lemma}
\begin{proof}
	The proof is an application of Proposition 7.29 / Proposition 7.46 / 
	Proposition 7.48 in \cite{bertsekas1978stochastic},
	depending on the given measurability. 
	Indeed, using the notation of \cite{bertsekas1978stochastic}, define the Borel-spaces
	$\mathcal{X}:= V\times\mathfrak{P}(W)$ and $\mathcal{Y}=W$
	as well as the Borel / upper semianalytic / universally measurable mapping 
	$f\colon \mathcal{X}\times\mathcal{Y}\to[-\infty,+\infty]$
	$f(x,y)= f(c,P,w):=X(v,P,w)$
	and Borel-kernel $q(dy,x)=q(dw,(v,P)):=P(dw)$.
	By the mentioned proposition the mapping
	$V\times\mathfrak{P}(W)=\mathcal{X}\to[-\infty,+\infty]$,
	$(v,P)=x\mapsto \int f(x,y)\,q(dy,x)=E_P[X(v,P,\cdot)]$
	is Borel / upper semianalytic / universally measurable.
\end{proof}

\begin{remark}
	\label{rem:kernels.measurable.and.derivative.measurable}
	By means of the disintegration theorem, every probability
	$P\in\mathfrak{P}(V\times W)$ can be written as $P=\mu\otimes K$, 
	where $\mu\in\mathfrak{P}(V)$ and $K\colon V\to\mathfrak{P}(W)$ is Borel.
	In fact, it is possible to construct the kernel $K$ in a way such that the mapping
	\[V\times\mathfrak{P}(V\times W)\to\mathfrak{P}(W),
	\quad (v,P)\mapsto K(v)\]
	is Borel; see \cite[Proposition 7.27]{bertsekas1978stochastic}.
\end{remark}

\bibliographystyle{abbrv}

\end{document}